\documentclass[journal]{IEEEtran}
\ifCLASSINFOpdf
\else
\fi

\usepackage{amsmath} 
\usepackage{amsthm}
\usepackage{amssymb}  
\usepackage[nospace,noadjust]{cite}
\usepackage{nicefrac}
\usepackage{graphicx}     
\usepackage{hhline}
\usepackage{amsbsy}
\usepackage{soul}
\usepackage{units}
\usepackage{relsize}
\usepackage{siunitx}
\usepackage{mathtools}
\usepackage{mathrsfs}
\usepackage{enumerate}
\usepackage[shortlabels]{enumitem}
\usepackage{tikzscale}
\usepackage{multirow}
\usepackage{diagbox}
\usepackage{microtype}
\usepackage{pgfplots}
\usepackage{tikz}
\usepackage{booktabs}
\usepgfplotslibrary{statistics}
\usepackage{pgfplotstable}
\usetikzlibrary{calc}
\usetikzlibrary{plotmarks}
\usetikzlibrary{patterns}
\usetikzlibrary{positioning}
\newlength\figureheight
\newlength\figurewidth
\newcommand{\abs}[1]{\left\lvert#1\right\rvert}
\newcommand{\norm}[1]{\left\lVert#1\right\rVert}
\newcommand{\nul}{\text{o}}
\newcommand{\wpone}{\text{ as } N\rightarrow\infty, \text{ w.p.1}}

\newtheorem{assumption}{Assumption}
\newtheorem{definition}{Definition}
\newtheorem{theorem}{Theorem}
\newtheorem{lemma}{Lemma}
\newtheorem{prop}{Proposition}
\newcommand{\thetals}{\hat{\theta}^\text{LS}_N}
\newcommand{\thetawls}{\hat{\theta}^\text{WLS}_N}

\newcommand{\minus}{\scalebox{0.5}[1.0]{$-$}}
\pgfplotsset{every tick label/.append style={font=\scriptsize}}
\mathtoolsset{showonlyrefs}

\begin{document}
\title{Parametric Identification Using \\ Weighted Null-Space Fitting}
\author{Miguel~Galrinho,
        Cristian~R.~Rojas,~\IEEEmembership{Member,~IEEE,}
        and~H{\aa}kan~Hjalmarsson,~\IEEEmembership{Fellow,~IEEE}
\thanks{Automatic Control Lab and ACCESS Linnaeus Center, School of Electrical
    Engineering, KTH Royal Institute of Technology, SE-100 44 Stockholm,
    Sweden. (e-mail: \tt{\{galrinho, crro, hjalmars\}}@kth.se.)}%
\thanks{This work was supported by the Swedish Research Council under contracts 2015-05285 and 2016-06079.}
}

\markboth{IEEE Transactions on Automatic Control}%
{Galrinho \MakeLowercase{\textit{et al.}}: The Weighted Null-Space Fitting Method}

\maketitle

\begin{abstract}
In identification of dynamical systems, the prediction error method using a quadratic cost function provides asymptotically efficient estimates under Gaussian noise and additional mild assumptions, but in general it requires solving a non-convex optimization problem.
An alternative class of methods uses a non-parametric model as intermediate step to obtain the model of interest.
Weighted null-space fitting (WNSF) belongs to this class.
It is a weighted least-squares method consisting of three steps.
In the first step, a high-order ARX model is estimated.
In a second least-squares step, this high-order estimate is reduced to a parametric estimate.
In the third step, weighted least squares is used to reduce the variance of the estimates.
The method is flexible in parametrization and suitable for both open- and closed-loop data.
In this paper, we show that WNSF provides estimates with the same asymptotic properties as PEM with a quadratic cost function when the model orders are chosen according to the true system.
Also, simulation studies indicate that WNSF may be competitive with state-of-the-art methods.
\end{abstract}


\IEEEpeerreviewmaketitle

\section{Introduction}

For parametric system identification, the prediction error method
(PEM) is the reference in the field. With open-loop data, consistency is
guaranteed if the model can describe the system dynamics, irrespective
of the used noise model. 
For Gaussian noise\footnote{When maximum likelihood and asymptotic efficiency are discussed in the following, the standard assumption is that the noise is Gaussian.} and with a noise model able to describe the noise spectrum, PEM with a quadratic cost function corresponds to maximum likelihood (ML), and is asymptotically efficient with respect to the used model structure~\cite{ljung99}, meaning that the covariance of the estimate asymptotically achieves the Cram{\'e}r-Rao (CR) lower bound: the best possible covariance for consistent estimators.
 
There are two issues that may hinder successful application of PEM. The first---and most
critical---is the risk of converging to a
non-global minimum of the cost function, which is in general not
convex. Thus, PEM requires local non-linear optimization algorithms and good initialization points.
The second issue concerns closed-loop data. 
In this case, PEM is biased unless the noise model is flexible enough.
For asymptotic efficiency, the noise model must be of correct order and estimated simultaneously with the dynamic model.
 
During the half decade since the publication of \cite{astrom&bohlin:65},
alternatives to PEM/ML have appeared, addressing one or both of
the aforementioned issues. We will not attempt to fully review this vast
field, but below we highlight some of the milestones.
 
Instrumental variable (IV) methods~\cite{soderstrom1983instrumental} allow consistency to be obtained in a large variety of settings without the issue of non-convexity.
Asymptotic efficiency can be obtained for some problems using iterative algorithms~\cite{Stoica&Soderstrom:83,riv}.
However, IV methods cannot achieve the CR bound with closed-loop data~\cite{cliv}.
 
Realization-based methods \cite{Kung:1978}, which later evolved into subspace methods \cite{subspacebook}, are non-iterative and thus attractive for their computational efficiency. 
The bias issue for closed-loop data has been overcome by more recent algorithms~\cite{verhaegen:closedloop,qin2003closed,jansson03,chiuso2005consistency}.
However, structural information is difficult to incorporate, and---even if a complete analysis is still unavailable (significant contributions have been provided~\cite{janwah96, Bauer05,chiuso2005consistency})---subspace methods are in general not believed to be as accurate as PEM.
 
Some methods are based on fixing some parameters in certain places of the cost function but not others to obtain a quadratic cost function, so that the estimate can be obtained by (weighted) least squares.
Then, the fixed coefficients are replaced by an estimate from the previous iteration in the weighting or in a filtering step.
This leads to iterative methods, which date back to~\cite{sanathanankoerner}.
Some of these methods have been denoted iterative quadratic maximum likelihood (IQML), originally developed for filter design~\cite{evansfischl73,bresler86} and later applied to dynamical systems~\cite{shaw94,shawmisra94,lemmerlingettal01}.
Another classical example is the Steiglitz-McBride method~\cite{stmcb_original} for estimating output-error models, which is equivalent to IQML for an impulse-input case~\cite{mcclellanlee91}.
In the identification field, weightings or filterings have not been determined by statistical considerations.
In this perspective, the result in~\cite{stoica1981steiglitz}, showing that the Steiglitz-McBride method is not asymptotically efficient, is not surprising.
 
Another approach is to estimate, in an intermediate step,
a more flexible model, followed by a model reduction step
to recover a model with the desired structure. The motivation for
this procedure is that, in some cases, each step corresponds to a
convex optimization problem or a numerically reliable procedure. To guarantee asymptotic efficiency,
it is important that the intermediate model is a
sufficient statistic and the
model reduction step is performed in a statistically sound way. Indirect
PEM~\cite{ipem_ss91} formalizes the requirements starting with an
over-parametrized model of fixed order and uses ML
in the model reduction step. The latter step corresponds in general to
a weighted non-linear least-squares problem. 
 
It has also been recognized that the intermediate model does not need to capture the true system perfectly, but only with sufficient accuracy.
Subspace algorithms can be interpreted in this way: for example, SSARX~\cite{jansson03} estimates an ARX model followed by a singular-value-decomposition (SVD) step and least-squares estimation. 
For spectral estimation, the THREE-like approach~\cite{byrnes2000three} is also a two-step procedure that first obtains a non-parametric spectral estimate and then reduces it to a parametric estimate that in some sense is closest to the non-parametric one, and whose optimization function is convex~\cite{zorzi2015three}.

In the field of time-series estimation, methods based on an intermediate high-order time series have also been suggested as alternative to ML, whose properties were studied in~\cite{hannan1973asymptotic}.
Durbin's first method~\cite{durbin1960fitting} for auto-regressive moving-average (ARMA) time series uses an intermediate high-order AR time series to simulate the innovations sequence, which allows obtaining the ARMA parameters by least squares.
This method does not have the same asymptotic properties as ML, unlike Durbin's second method~\cite{durbin1960fitting}.
The latter is an extension of~\cite{durbin1959efficient}, which had been proposed for MA time series, whose parameters can be estimated from the high-order AR estimates using least squares, with an accuracy that can be made arbitrarily close to the Cram{\'e}r-Rao bound by increasing the AR-model order.
When applied to ARMA time series, the idea to achieve efficiency is to iterate between estimating the AR and MA polynomials using this procedure, initialized with Durbin's first method.
Another way to achieve efficiency from Durbin's first method as starting point was proposed in~\cite{mayne1982linear} by using an additional filtering step with the MA estimates from Durbin's first method, and then re-estimating the ARMA-parameters.

The asymptotic properties of these methods have been analyzed by considering the high order tending to infinity, but ``small'' compared to the sample size.
A preferable analysis should handle the relation between the high order and the sample size formally, as done in~\cite{hannanKavalieris} to prove consistency of the method in~\cite{mayne1982linear}, where the high order is assumed to tend to infinity as function of the sample size at a particular rate.
This class of methods has become popular for vector ARMA time series, with several available algorithms using different procedures for obtaining the asymptotic efficient model parameter estimates from the estimated innovations (e.g., \cite{hannan1984multivariate,reinsel1992maximum,dufour2014asymptotic}, with further information in references in these papers).
Despite sharing the same asymptotic properties, which have been analyzed with the high order as a function of the sample size, these algorithms may have different computational requirements and finite sample properties.

For identification of dynamical systems, instead of using the high-order model to estimate the innovations, it has been suggested that identification of the model of interest can be done by applying asymptotic ML directly to the high-order model~\cite{wahlberg89}. 
The ASYM method~\cite{zhu_book} is an instantiation of this approach. 
Because an ARX-model estimate and its covariance constitute a sufficient statistic as the model order grows, this approach can produce asymptotically efficient estimates. 
However, the plant and noise models are estimated separately, preventing asymptotic efficiency for closed-loop data. 
Also, although such model reduction procedures may have numerical advantages over direct application of PEM~\cite{zhu_book}, this approach still requires local non-linear optimization techniques.
The Box-Jenkins Steiglitz-McBride (BJSM) method~\cite{bjsm} instead uses the Steiglitz-McBride method in the model reduction step, resulting in asymptotically efficient estimates of the plant in open loop. 
Two drawbacks of BJSM are that the number of iterations is required to tend to infinity (as for the Steiglitz-McBride method) and that, similarly to \cite{wahlberg89} and \cite{zhu_book}, the CR bound cannot be attained in closed loop.
The Model Order Reduction Steiglitz-McBride (MORSM) method solves the first drawback of BJSM, but not the second~\cite{MORSM}.
 
In this contribution, we focus on weighted null-space fitting (WNSF),
introduced in~\cite{galrinho14}. This method uses two of the features of the methods above: i) an intermediate high-order
ARX model; ii) the high-order model is directly used for estimating the low-order model using ML-based model reduction. 
However, instead of an explicit minimization of the model-reduction cost function---as in indirect PEM (directly via the
model parameters), ASYM (in the time domain), and \cite{wahlberg89} (in the frequency domain)---the model reduction
step consists of a weighted least-squares problem. 
Asymptotic efficiency requires that the weighting depend on the (to be estimated) model parameters. 
To handle this, an additional
least-squares step is introduced. 
Consisting of three (weighted) least-squares steps, WNSF has attractive computational properties in comparison with, for example, PEM, ASYM, and BJSM. 
More steps may be added to this standard procedure, using an iterative weighted least-squares algorithm, which may improve the estimate for finite sample size.

Another interesting feature of WNSF is that,
unlike many of the methods above (including MORSM), the dynamic model and the noise model are estimated jointly.
If this is not done, an algorithm cannot be asymptotically efficient for closed-loop data~\cite{Forssell_cl}.
Nevertheless, in some applications, the noise model may be of no concern. 
WNSF can then be simplified and a noise model not estimated, still maintaining asymptotic efficiency for open-loop data. 
In closed loop, consistency is still maintained because the high-order model captures the noise spectrum consistently, while the resulting accuracy corresponds to the covariance of PEM with an infinite-order noise model~\cite{Forssell_cl}. 
Thus, besides the attractive numerical properties, WNSF has
theoretical properties matched only by PEM. 
However, WNSF has the additional benefit that an explicit noise model is not required to obtain consistency with closed-loop data.
 
In~\cite{galrinho14}, some theoretical properties of WNSF are claimed and supported by simulations, but with no formal proof.
The robust performance that the method has shown has provided the motivation to extend the simulation study and deepen the theoretical analysis.
Take Fig.~\ref{fig:peakssys} as an example, showing the FITs (see~\eqref{eq:fit} for a definition of this quality measure) 	for estimates obtained in closed loop from 100 Monte Carlo runs with the following methods: PEM with default MATLAB implementation (PEM-d), the subspace method SSARX~\cite{jansson03}, WNSF, and PEM initialized at the true parameters as benchmark (PEM-t).
Here, the default MATLAB initialization for PEM is often not accurate enough, and the non-convex cost function of PEM converges to non-global minima, while the low FIT of SSARX indicates that this method is not a viable alternative to deal with the non-convexity of PEM for the situation at hand.
On the other hand, WNSF has a performance close to PEM initialized at the true parameters, suggesting that the weighted least-squares procedure applied to a non-parametric estimate may be more robust than an explicit minimization of the PEM cost function in regards to convergence issues.

In this paper, we provide a theoretical and experimental analysis of WNSF applied to stable single-input single-output (SISO) Box-Jenkins (BJ) systems, which may operate in closed loop. 
Our main contributions are to establish conditions for consistency and asymptotic efficiency.
A major effort of the analysis is to keep track of the model errors induced by using an ARX model on data generated by a system of BJ type. 
It is a delicate matter to determine how the ARX-model
order should depend on the sample size such that it is
ensured that these errors vanish as the sample size grows: to this end, the results in \cite{ljung&wahlberg92} have been instrumental.
We finally conduct a finite sample simulation study where WNSF shows competitive performance with state-of-the-art methods.

The paper is organized as follows.
In Section~\ref{sec:prel}, we introduce definitions, assumptions, and background.
In Section~\ref{sec:wnsf}, we review the WNSF algorithm.
In Section~\ref{sec:prop}, we provide the theoretical analysis;
in Section~\ref{sec:sim}, the experimental analysis.

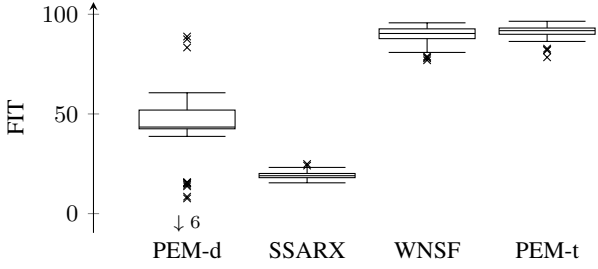
\begin{figure}
\pgfplotsset{compat=1.14}
\begin{tikzpicture}
  \begin{axis}[
    boxplot/draw direction=y,
    width=\figurewidth,
    height=.85\figureheight,
    ymin  = 0,
    x axis line style={opacity=0},
    axis x line*=bottom,
    axis y line=left,
    enlarge y limits,
    xtick style={draw=none},
    xticklabels={,,PEM-d,SSARX,WNSF,PEM-t},
    ylabel=FIT,
    label style={font=\small},
    tick label style={font=\small},
  ]
    \addplot [boxplot,mark=x] table[y index=0]{Data/peakssys_pem.csv};
    \node at (axis cs:1,4) [anchor=north] {\scriptsize{$\downarrow 6$}};
    \addplot [boxplot,mark=x] table[y index=0]{Data/peakssys_ssarx.csv};
    \addplot [boxplot,mark=x] table[y index=0]{Data/peakssys_wnsf.csv};
    \addplot [boxplot,mark=x] table[y index=0]{Data/peakssys_pem_true.csv};
  \end{axis}
\end{tikzpicture}
\caption{FITs from 100 Monte Carlo runs with a highly resonant system.}
\label{fig:peakssys}
\end{figure}

\section{Preliminaries}
\label{sec:prel}

\subsection{Notation}
\begin{itemize}
\item $||x||_p=(\sum_{k=1}^{n}\abs{x_k}^p)^{1/p}$, with $x_k$ the $k^\text{th}$ entry of the $n\times 1$ vector $x$, and $p\in\mathbb{N}$ (for simplicity $||x||:=||x||_2$).
\item $||A||_p=\sup_{x\neq 0} ||Ax||_p/||x||_p$, with $A$ a matrix, $x$ a vector of appropriate dimensions, and $p\in\mathbb{N}$ (for simplicity $||A||:=||A||_2$); also, $||A||_\infty = ||A^\top||_1$.
\item $C$ and $\bar{N}$ denote any constant, which need not be the same in different expressions, and may be random variables. 
\item $\Gamma_n(q) = [q^{-1} \quad \cdots q^{-n}]^\top$, where $q^{-1}$ is the backward time-shift operator.
\item $A^*$ is the complex conjugate transpose of the matrix $A$.
\item $\mathcal{T}_{n,m}(X(q))$ is the Toeplitz matrix of size $n\times m$ ($m\leq n$) with first column $[x_0 \; \cdots \; x_{n-1}]^\top$ and first row $[x_0 \; 0_{1\times m-1}]$, where $X(q)=\sum_{k=0}^\infty x_k q^{\minus k}$. The dimension $n$ may be infinity, denoted $\mathcal{T}_{\infty,m}(X(q))$.
\item $\mathbb{E} x$ denotes expectation of the random vector $x$.
\item $\bar{\mathbb{E}} x_t := \lim\limits_{N\to\infty} \frac{1}{N} \sum_{t=1}^{N} \mathbb{E} x_t$.
\item $x_N = \mathcal{O}(f_N)$: the function $x_N$ tends to zero at a rate not slower than $f_N$, as $N\to\infty$, w.p.1.
\item $x_N \sim As \mathcal{N} (a,P)$: the random variable $x_N$ is normally distributed with mean $a$ and covariance $P$ as $N\to\infty$.
\end{itemize}

\subsection{Definitions and Assumptions}

\begin{assumption}[Model and true system]
\label{ass:model_truesystem}
The model has input $\{u_t\}$, output $\{y_t\}$ and is subject to the noise $\{e_t\}$, all real-valued, related by
\begin{equation}
y_t = G(q,\theta)u_t + H(q,\theta)e_t .
\label{eq:model}
\end{equation}
The transfer functions $G(q,\theta)$ and $H(q,\theta)$ are rational functions in $q^{\minus 1}$, according to
\begin{alignat}{5}
& G(q,\theta) && := && \frac{L(q,\theta)}{F(q,\theta)} && := && \frac{l_1 q^{\minus 1} + \dots + l_{m_l} q^{\minus m_l}}{1 + f_1 q^{\minus 1} + \dots + f_{m_f} q^{\minus m_f}} , \\
& H(q,\theta) && := && \frac{C(q,\theta)}{D(q,\theta)} && := && \frac{1 + c_1 q^{\minus 1} + \dots + c_{m_c} q^{\minus m_c}}{1 + d_1 q^{\minus 1} + \dots + d_{m_d} q^{\minus m_d}} ,
\end{alignat}
where $\theta$ is the parameter vector to be estimated, given by
\begin{equation}
\theta = 
\begin{bmatrix}
f^\top & l^\top & c^\top & d^\top
\end{bmatrix}^\top \in \mathbb{R}^{m_f+m_l+m_c+m_d},
\label{eq:theta}
\end{equation}
with
\begin{equation}
f = 
\begin{bmatrix}
f_1 \\ \vdots \\ f_{m_f}
\end{bmatrix} , \quad
l = 
\begin{bmatrix}
l_1 \\ \vdots \\ l_{m_l}
\end{bmatrix} , \quad
c = 
\begin{bmatrix}
c_1 \\ \vdots \\ c_{m_c}
\end{bmatrix} , \quad
d = 
\begin{bmatrix}
d_1 \\ \vdots \\ d_{m_d}
\end{bmatrix} .
\label{eq:flc}
\end{equation}
If the noise model is not of interest, we consider that we want to obtain an estimate $G(q,\bar{\theta})$, where
$\bar{\theta} = [f^\top \quad l^\top]^\top$.

The true system is described by~\eqref{eq:model} when $\theta=\theta_\nul$. The transfer functions ${G_\nul:=G(q,\theta_\nul)}$ and ${H_\nul:=H(q,\theta_\nul)}$ are assumed to be stable, and $H_\nul$ inversely stable. The polynomials $L_\nul:=L(q,\theta_\nul)$ and $F_\nul:=F(q,\theta_\nul)$, as well as $C_\nul:=C(q,\theta_\nul)$ and $D_\nul:=D(q,\theta_\nul)$, do not share common factors.
\end{assumption}


Because we allow for data to be collected in closed loop, the input $\{u_t\}$ is allowed to have a stochastic part.
Then, let $\mathcal{F}_{t-1}$ be the $\sigma$-algebra generated by ${\{e_s, u_s, s\leq t-1\}}$. 
For the noise, the following assumption applies.

\begin{assumption}[Noise]
\label{ass:noise}
The noise sequence $\{e_t\}$ is a stochastic process that satisfies
\begin{equation}
\mathbb{E}[e_t|\mathcal{F}_{t-1}] = 0, \quad  \mathbb{E}[e_t^2|\mathcal{F}_{t-1}] = \sigma_\nul^2, \quad \mathbb{E}[|e_t|^{10}] \leq C, \forall t .
\end{equation}
\end{assumption}

Before stating the assumption on the input sequence, we introduce the following definitions, used in~\cite{ljung&wahlberg92}.

\begin{definition}[$f_N$-quasi-stationarity]
Let $f_N$ be a decreasing sequence of positive scalars, with $f_N \to 0$ as $N \to \infty$, and
\begin{equation}
R^N_{vv}(\tau) = \left \{ 
\begin{array}{lr}
\frac{1}{N}\sum_{t = \tau+1}^N v_tv_{t-\tau}^\top, & 0 \le \tau < N, \\
\frac{1}{N}\sum_{t = 1}^{N+\tau} v_tv_{t-\tau}^\top, & -N < \tau \leq 0,\\
0, & \mathrm{otherwise.}
\end{array}
\right.
\end{equation}
The vector sequence $\{v_t\}$ is $f_N$-quasi-stationary if
\begin{enumerate}[i)]
\item There exists $R_{vv}(\tau)$ such that \\ $\sup_{|\tau| \leq N} \norm{R^N_{vv}(\tau)- R_{vv}(\tau)} \leq C_1 f_N$,
\item $\frac{1}{N}\sum_{t = -N}^N \norm{v_t}^2  \le C_2$
\end{enumerate}
for all $N$ large enough, where $C_1$ and $C_2$ are finite constants.
\end{definition}
This definition allows us to work with some stochastic signals that have deterministic components, as in~\cite{ljung99}.
In addition to the standard definition of quasi-stationarity, a rate of convergence for the sample covariances is defined.

\begin{definition}[$f_N$-stability]
A filter $G(q)  =  \sum_{k=0}^\infty  g_k q^{\minus k}$
is $f_N$-stable if $\sum_{k=0}^\infty |g_k| /f_k < \infty $.
\end{definition}

\begin{definition}[Power spectral density]
The power spectral density of an $f_N$-quasi-stationary sequence $\{v_t\}$ is given by $\Phi_{v}(z)=\sum_{\tau=\minus\infty}^\infty R_{vv}(\tau)z^{-\tau}$, if the sum exists for $|z|=1$.
\end{definition}

For the input, the following assumption applies.

\begin{assumption}[Input]
\label{ass:input}
The input sequence $\{u_t\}$ is defined by
$u_t = - K(q) y_t + r_t$
under the following conditions.
\begin{enumerate}[i)]
\item The sequence $\{r_t\}$ is independent of $\{e_t\}$, $f_N$-quasi-stationary with $f_N\!=\!\sqrt{\log N/N}$, and uniformly bounded.
\item With $\Phi_{r}(z) = F_r(z)F_r(z^{\minus 1})$ the spectral factorization of $\{r_t\}$ and $F_r(z)$ causal, $F_r(q)$ is BIBO stable.
\item The closed loop system is $f_N$-stable with $f_N = 1/\sqrt{N}$.
\item The transfer function $K(z)$ is bounded on the unit circle.
\item The spectral density of $\{[ r_t \; e_t ]^\top \}$ is coercive (i.e., bounded from below by the matrix $\delta I$, for some $\delta>0$).
\end{enumerate}
\end{assumption}
Operation in open loop is obtained by taking $K(q)=0$. 
The choice of $f_N$ in \emph{iii)} guarantees that the impulse responses of the closed-loop system have a minimum rate of decay, necessary to derive the results in~\cite{ljung&wahlberg92}.
This minimum decay rate is trivially satisfied here, as the system is stable and finite dimensional, and hence has exponentially decaying impulse responses.

\subsection{The Prediction Error Method}

The prediction error method minimizes a cost function of prediction errors, which, for the model structure~\eqref{eq:model}, are
\begin{equation}
\varepsilon_t(\theta) = \frac{D(q,\theta)}{C(q,\theta)} \left( y_t - \frac{L(q,\theta)}{F(q,\theta)} u_t \right) .
\label{eq:epsilon}
\end{equation}
Using a quadratic cost function, the PEM estimate of $\theta$ is obtained by minimizing
\begin{equation}
\textstyle
J(\theta) = \frac{1}{N} \sum_{t=1}^{N} \varepsilon_t^2(\theta) ,
\label{eq:J}
\end{equation}
where $N$ is the sample size.
Assuming that $\theta$ belongs to an appropriate domain~\cite[Def. 4.3]{ljung99}, when the data set is informative~\cite[Def. 8.1]{ljung99} and under appropriate technical conditions~\cite[Chap. 8]{ljung99}, the global minimizer $\hat{\theta}_N^\text{PEM}$ of~\eqref{eq:J} is asymptotically distributed as~\cite[Theorem 9.1]{ljung99}
\begin{equation}
\sqrt{N} (\hat{\theta}_N^\text{PEM}-\theta_\nul) \sim As \mathcal{N} \big(0,\sigma_\nul^2 M_\text{CR}^{\minus 1}\big),
\label{eq:PEMcov}
\end{equation}
where
\begin{equation}
M_\text{CR} = \frac{1}{2\pi} \int_{-\pi}^{\pi} \Omega(e^{i\omega}) \Phi_z(e^{i\omega}) \Omega^*(e^{i\omega}) d\omega ,
\label{eq:CR}
\end{equation}
with (for notational simplicity, we omit the argument $e^{i\omega}$)
\begin{equation}
\Omega = 
\begin{bmatrix}
-\frac{G_\nul}{H_\nul F_\nul} \Gamma_{m_f} & 0\\
\frac{1}{H_\nul F_\nul} \Gamma_{m_l} & 0 \\
0 & \frac{1}{C_\nul} \Gamma_{m_c} \\
0 & -\frac{1}{D_\nul} \Gamma_{m_d}
\end{bmatrix}
\label{eq:Omega}
\end{equation}
and $\Phi_z$ the spectrum of $\left[u_t \;\; e_t\right]^\top\!$.
When the error sequence is Gaussian, PEM with a quadratic cost function is asymptotically efficient, with~\eqref{eq:CR} corresponding to the CR bound~\cite[Chap. 9]{ljung99}.

In open loop, the asymptotic covariance of the dynamic-model parameters is the top-left block of~\eqref{eq:PEMcov} even if the noise-model orders $m_c$ and $m_d$ are larger than the true ones; if smaller, the dynamic-model estimates are consistent but not asymptotically efficient.
In closed loop, the covariance of the dynamic-model estimates only corresponds to the top-left block of~\eqref{eq:PEMcov} if the noise-model orders are the true ones; if smaller, the dynamic-model estimates are biased; if larger, they are consistent and the asymptotic covariance matrix can be bounded by $\sigma_\nul^2 M_\text{CL}^{-1}$, where~\cite{Forssell_cl}
\begin{equation}
M_\text{CL} = 
\frac{1}{2\pi} \int_{-\pi}^{\pi} \bar{\Omega}(e^{i\omega}) \Phi_u^r(e^{i\omega}) \bar{\Omega}^*(e^{i\omega}) d\omega 
\label{eq:Mcl}
\end{equation}
with $\Phi_u^r$ the spectrum of the input due to the reference.
This corresponds to the case with infinite noise-model order.

The main drawback with PEM is that minimizing~\eqref{eq:J} is in general a non-convex optimization problem.
Therefore, the global minimizer $\hat{\theta}_N^\text{PEM}$ is not guaranteed to be found.
An exception is the ARX model.

\subsection{High-Order ARX Modeling}

The true system can alternatively be written as
\begin{equation}
A_\nul(q) y_t = B_\nul(q) u_t + e_t ,
\label{eq:truearx}
\end{equation}
where the transfer functions
\begin{equation}
\begin{alignedat}{3}
A_\nul(q) &:= \frac{1}{H_\nul(q)} &&=: 1+&&\sum_{k=1}^{\infty} a^\nul_k q^{\minus k} , \\
B_\nul(q) &:= \frac{G_\nul(q)}{H_\nul(q)} &&=: &&\sum_{k=1}^{\infty} b^\nul_k q^{\minus k} 
\end{alignedat}
\label{eq:truearxpoly}
\end{equation}
are stable (Assumption~\ref{ass:model_truesystem}).
Therefore, the ARX model
\begin{equation}
A(q,\eta^n) y_t = B(q,\eta^n) u_t + e_t ,
\label{eq:arxmodel}
\end{equation}
where
\begin{equation}
\eta^n = 
\begin{bmatrix}
a_1 & \cdots & a_n & b_1 & \cdots & b_n
\end{bmatrix}^\top ,
\end{equation}
\begin{equation}
A(q,\eta^n) = 1+\sum_{k=1}^{n} a_k q^{\minus k} , \quad B(q,\eta^n) = \sum_{k=1}^{n} b_k q^{\minus k} ,
\end{equation}
can approximate~\eqref{eq:truearx} arbitrarily well if the model order $n$ is chosen arbitrarily large.

Because the prediction errors for the ARX model~\eqref{eq:arxmodel},
$
\varepsilon_t(\eta^n) = A(q,\eta^n) y_t - B(q,\eta^n) u_t ,
$
are linear in the model parameters $\eta^n$, the corresponding PEM cost function~\eqref{eq:J} can be minimized with least squares.
This is done as follows.
First, re-write~\eqref{eq:arxmodel} in regression form as
\begin{equation}
y_t = (\varphi_t^n)^\top \eta^n + e_t ,
\label{eq:ARXregression}
\end{equation}
where
\begin{equation}
\varphi_t^n =
\begin{bmatrix}
-y_{t-1} & \cdots & - y_{t-n} & u_{t-1} & \cdots & u_{t-n}
\end{bmatrix}^\top.
\label{eq:phi}
\end{equation}
Then, the least-squares estimate of $\eta^n$ is obtained by
\begin{equation}
\hat{\eta}^{n,\text{ls}}_N = [R^n_N]^{\minus 1} r^n_N ,
\label{eq:eta_ls}
\end{equation}
where
\begin{equation}
R^n_N = \frac{1}{N} \sum_{t=n+1}^N \varphi_t^n (\varphi_t^n)^\top , \quad
r^n_N = \frac{1}{N} \sum_{t=n+1}^N \varphi_t^n y_t .
\label{eq:RnN}
\end{equation}
As the sample size increases, we have~\cite{ljung&wahlberg92}
\begin{equation}
\begin{alignedat}{3}
R^n_N &\to \bar{R}^n&& \left(:=\bar{\mathbb{E}}\left[ \varphi_t^n (\varphi_t^n)^\top \right]\right) , &&\wpone , \\
r^n_N &\to \bar{r}^n&& \left(:=\bar{\mathbb{E}}\left[ \varphi_t^n y_t \right] \right) , &&\wpone .
\end{alignedat}
\label{eq:Rbar}
\end{equation}
Consequently,
\begin{equation}
\hat{\eta}^{n,\text{ls}}_N \to \bar{\eta}^n := \left[\bar{R}^n\right]^{\minus 1} \bar{r}^n, \wpone .
\label{eq:etahat_conv}
\end{equation}
For future reference, we define
\begin{equation}
\begin{aligned}
\eta^n_\nul &:= 
\begin{bmatrix}
a_1^\nul & \cdots & a_n^\nul & b_1^\nul & \cdots & b_n^\nul
\end{bmatrix}^\top , \\
\eta_\nul &:= 
\begin{bmatrix}
a_1^\nul & a_2^\nul & \cdots & b_1^\nul & b_2^\nul & \cdots
\end{bmatrix}^\top .
\end{aligned}
\label{eq:eta0}
\end{equation}

The attractiveness of ARX modeling is the simplicity of estimation while approximating more general classes of systems with arbitrary accuracy.
However, as the order $n$ typically has to be taken large, the estimated ARX model will have high variance.
Nevertheless, this estimate can in principle be used as a means to obtain an asymptotically estimate of the Box-Jenkins (BJ) model~\eqref{eq:model} when the measurement noise is Gaussian.
To understand intuitively why this is the case, we observe that if we neglect the truncation error from approximating the true system~\eqref{eq:truearx} by the model~\eqref{eq:arxmodel}, $\hat{\eta}^{n,\text{ls}}_N$ and $R_N^n$ constitute a sufficient statistic for our problem.
Therefore, they can replace the data without loss of information.
If ML is used for the subsequent estimation, we need to solve a non-convex optimization problem~\cite{wahlberg89}.
An accurate estimate to initialize the optimization procedure is then crucial; a standard result is that, if initialized with a strongly consistent estimate, one Gauss-Newton iteration provides an asymptotically efficient estimate (e.g., \cite[Chap. 23]{gourieroux1995statistics}).

\section{Weighted Null-Space Fitting Method}
\label{sec:wnsf}

The idea of weighted null-space fitting~\cite{galrinho14} is to avoid the burden of a non-convex optimization by using weighted least squares, but maintaining the properties of maximum likelihood.
The method consists of three steps.
In the first step, a high-order ARX model is estimated with least squares.
In the second step, the parametric model is estimated from the high-order ARX model with least squares, providing a consistent estimate.
In the third step, the parametric model is re-estimated with weighted least squares.
Because the optimal weighting depends on the true parameters, we replace these by the consistent estimate obtained in the previous step.
Similarly to maximum likelihood with an optimization algorithm initialized at a consistent estimate, this provides an asymptotically efficient estimate.
We now proceed to detail each of these steps.

The first step consists in estimating $\hat{\eta}^{n,\text{ls}}_N$ from~\eqref{eq:eta_ls}.
As discussed before, $\hat{\eta}^{n,\text{ls}}_N$ and $R^n_N$ are almost a sufficient statistic for our problem, if the ARX-model truncation error is small enough (later, this will be treated formally).
Then, we will use $\hat{\eta}^{n,\text{ls}}_N$ and $R^n_N$ instead of data to estimate the model of interest.

The second step implements this as follows.
Re-write~\eqref{eq:truearxpoly} as
\begin{equation}
\begin{aligned}
C_\nul(q)A_\nul(q) - D_\nul(q) &= 0  , \\
F_\nul(q)B_\nul(q) - L_\nul(q)A_\nul(q) &= 0 .
\end{aligned}
\label{eq:CA-D=0andFB-LA=0} 
\end{equation}
Then, \eqref{eq:CA-D=0andFB-LA=0} can be expanded as
\begin{subequations}
\begin{equation}
\begin{multlined}
(1+c_1^\nul q^{\minus 1}+\cdots+c_{m_c}^\nul q^{\minus m_c}) \Big(1+\sum_{k=1}^\infty a_k^\nul q^{\minus k}\Big) \\
- (1+d_1^\nul q^{\minus 1}+\cdots+d_{m_d}^\nul q^{\minus m_d}) = 0 ,
\end{multlined} \label{eq:CA-D=0_exp}
\end{equation}
\begin{equation}
\begin{multlined}
(1+f_1^\nul q^{\minus 1}+\cdots+f_{m_f}^\nul q^{\minus m_f}) \sum_{k=1}^\infty b_k^\nul q^{\minus k} \\
 - (l_1^\nul q^{\minus 1}+\cdots+l^\nul_{m_l}q^{\minus m_l}) \Big(1+\sum_{k=1}^\infty a_k^\nul q^{\minus k}\Big) = 0 . 
\end{multlined} \label{eq:FB-LA=0_exp}
\end{equation}
\label{eq:CA-D=0andFB-LA=0_exp}
\end{subequations}
To express $\theta_\nul$ in terms of $\eta_\nul$, we can do so in vector form.
Because a power-series product (e.g., $ \sum_{k=0}^\infty \alpha_k q^{\minus k} \sum_{k=0}^\infty \beta_k q^{\minus k} $) can be written as the (Toeplitz-)matrix-vector product
\begin{equation}
\begin{bmatrix}
\begin{matrix}
\alpha_0 &         \\
\alpha_1 & \alpha_0
\end{matrix} \!\! &
\begin{matrix}
\mathlarger{ \mathlarger{ 0 } }
\end{matrix} \\
\begin{matrix}
\alpha_2 & \alpha_1 \\
\vdots   & \ddots
\end{matrix} \!\! & 
\begin{matrix}
\alpha_0 &   \\
\ddots   & \ddots
\end{matrix}
\end{bmatrix}
\begin{bmatrix}
\beta_0 \\ \beta_1 \\ \beta_2 \\ \vdots
\end{bmatrix}=
\begin{bmatrix}
\begin{matrix}
\beta_0 &         \\
\beta_1 & \beta_0
\end{matrix} \!\! &
\begin{matrix}
\mathlarger{ \mathlarger{ 0 } }
\end{matrix} \\
\begin{matrix}
\beta_2 & \beta_1 \\
\vdots   & \ddots
\end{matrix} \!\!& 
\begin{matrix}
\beta_0 &   \\
\ddots   & \ddots
\end{matrix}
\end{bmatrix}
\begin{bmatrix}
\alpha_0 \\ \alpha_1 \\ \alpha_2 \\ \vdots
\end{bmatrix},
\label{eq:Toeplitzform}
\end{equation}
we may write~\eqref{eq:CA-D=0andFB-LA=0_exp} as (keeping the first $n$ equations)
\begin{equation}
\eta^n_\nul - Q_n(\eta^n_\nul) \theta_\nul = 0 ,
\label{eq:eta-Qtheta=0}
\end{equation}
with $\theta_\nul$ defined by~\eqref{eq:theta} evaluated at the true parameters and
\begin{equation}
Q_n(\eta^n) =
\begin{bmatrix}
0 & 0 & -Q^c_n(\eta^n) & Q^d_n \\
-Q^f_n(\eta^n) & Q^l_n(\eta^n) & 0 & 0
\end{bmatrix} ,
\label{eq:Qeta}
\end{equation}
where, when evaluated at the true parameters $\eta^n_\nul$,
\begin{equation}
\begin{aligned}
&Q^c_n(\eta^n_\nul) = \mathcal{T}_{n,m_c}(A(q,\eta_\nul)), \;\;\; 
Q^l_n(\eta^n_\nul) = \mathcal{T}_{n,m_l}(A(q,\eta_\nul)), \;\;\; \\
&Q^f_n(\eta^n_\nul) = \mathcal{T}_{n,m_f}(B(q,\eta_\nul)), \;\;\;
Q^d_n = 
\begin{bmatrix}
I_{m_d,m_d} \\ 0_{n-m_d,m_d}
\end{bmatrix} .
\end{aligned}
\label{eq:Q0blocksdef}
\end{equation}
Motivated by~\eqref{eq:eta-Qtheta=0}, we replace $\eta^n_\nul$ by its estimate $\hat{\eta}^{n,\text{ls}}_N$, obtaining an over-determined system of equations, which may be solved for $\theta$ using, for example, least squares:
\begin{equation}
\hat{\theta}_N^\text{LS} = \left( Q_n^\top(\hat{\eta}^{n,\text{ls}}_N) Q_n(\hat{\eta}^{n,\text{ls}}_N) \right)^{\minus 1} Q_n^\top(\hat{\eta}^{n,\text{ls}}_N) \hat{\eta}^{n,\text{ls}}_N .
\label{eq:theta_ls}
\end{equation}
In~\eqref{eq:theta_ls}, invertibility follows from convergence of $\hat{\eta}^{n,\text{ls}}_N$ to $\eta^n_\nul$, which is of larger dimension than $\theta$, and the block-Toeplitz structure of $Q_n(\eta^n)$ (this is treated formally in Lemma~\ref{thm:ls_inv}).

With~\eqref{eq:theta_ls}, we have not accounted for the residuals in~\eqref{eq:eta-Qtheta=0} when $\hat{\eta}^{n,\text{ls}}_N$ replaces $\eta^n_\nul$.
The third step remedies this by re-estimating $\theta$ in a statistically sound way.

For some $\eta^n$, and using the same logic as~\eqref{eq:Toeplitzform}, we can write~\eqref{eq:eta-Qtheta=0} as
\begin{equation}
\eta^n - Q_n(\eta^n) \theta_\nul = T_n(\theta_\nul) (\eta^n-\eta^n_\nul) =: \delta_n(\eta^n,\theta_\nul).
\label{eq:Tetadiff}
\end{equation}
where
\begin{equation}
T_n(\theta) = 
\begin{bmatrix}
T^c_n(\theta) & 0 \\
-T^l_n(\theta) & T^f_n(\theta)
\end{bmatrix},
\label{eq:T0}
\end{equation}
with $T^c_n(\theta_\nul) = \mathcal{T}_{n,n}(C(q,\theta_\nul))$, $T^l_n(\theta_\nul) = \mathcal{T}_{n,n}(L(q,\theta_\nul))$, and $T^f_n(\theta_\nul) = \mathcal{T}_{n,n}(F(q,\theta_\nul))$.
The objective is then to estimate $\theta$ that minimizes the residuals $\delta_n(\hat{\eta}^{n,\text{ls}}_N,\theta)$.
If we neglect the bias error from truncation of the ARX model, which should be close to zero for sufficiently large $n$, we have that, approximately,
\begin{equation}
\sqrt{N} \big(\hat{\eta}^{n,\text{ls}}_N-\eta^n_\nul\big) \sim As\mathcal{N}\big(0,\sigma_\nul^2 \left[\bar{R}^n\right]^{\minus 1} \big).
\label{eq:eta_dist}
\end{equation}
Then, using~\eqref{eq:eta_dist} and~\eqref{eq:Tetadiff}, we may write that, approximately,
\begin{equation}
\delta_n(\hat{\eta}^{n,\text{ls}}_N) \sim As\mathcal{N} \left(0,T_n(\theta_\nul)\sigma_\nul^2[\bar{R}^n]^{\minus 1}T^\top_n(\theta_\nul)\right).
\label{eq:Tetadiff_dist}
\end{equation}
Because the residuals we want to minimize, given by $\delta_n(\hat{\eta}^{n,\text{ls}}_N,\theta)=\hat{\eta}^{n,\text{ls}}_N-Q_n(\hat{\eta}^{n,\text{ls}}_N)\theta$, are asymptotically distributed by~\eqref{eq:Tetadiff_dist}, the estimate of $\theta$ with minimum variance is given by the weighted least-squares estimate
\begin{equation}
\hat{\theta}^{\text{WLS}_\nul}_N \! = \! \left( Q_n^\top(\hat{\eta}^{n,\text{ls}}_N) \bar{W}_n(\theta_\nul) Q_n(\hat{\eta}^{n,\text{ls}}_N) \right)^{\minus 1} \!\!\! Q_n^\top(\hat{\eta}^{n,\text{ls}}_N) \bar{W}_n(\theta_\nul) \hat{\eta}^{n,\text{ls}}_N ,
\end{equation}
where the weighting matrix
\begin{equation}
\bar{W}_n(\theta_\nul) = \left( T_n(\theta_\nul) \sigma_\nul^2[\bar{R}^n]^{\minus 1} T^\top_n(\theta_\nul) \right)^{\minus 1} 
\label{eq:barW}
\end{equation}
is the inverse of the covariance of the residuals~\cite{kailath:le}.
Because $\theta_\nul$ and $\bar{R}^n$ are not available, we replace them by $\hat{\theta}^{\text{LS}}_N$ and $R^n_N$, respectively ($\sigma_\nul^2$ can be disregarded, because the weighting can be scaled arbitrarily without influencing the solution).
Thus, the third step consists in re-estimating $\theta$ by
\begin{equation}
\hat{\theta}^{\text{WLS}}_N \!\! = \! \left(\! Q_n^\top(\hat{\eta}^{n,\text{ls}}_N) W_n(\hat{\theta}^{\text{LS}}_N) Q_n(\hat{\eta}^{n,\text{ls}}_N) \!\right)^{\minus 1} \!\!\!\! Q_n^\top(\hat{\eta}^{n,\text{ls}}_N) W_n(\hat{\theta}^{\text{LS}}_N) \hat{\eta}^{n,\text{ls}}_N , 
\label{eq:theta_wls}
\end{equation}
where (we take the inverses of the matrices individually)
\begin{equation}
W_n(\hat{\theta}^{\text{LS}}_N) 
= T_n^{\minus \top} (\hat{\theta}^{\text{LS}}_N) R^n_N T_n^{\minus 1} (\hat{\theta}^{\text{LS}}_N)  ,
\label{eq:Wls}
\end{equation}
with $T_n(\thetals)$ obtained using~\eqref{eq:T0}.
Invertibility in~\eqref{eq:theta_wls} follows from (besides what was mentioned for Step 2) invertibility of $W_n(\thetals)$, which in turn follows from the lower-Toeplitz structure of $T_n(\theta)$, convergence of $\thetals$ to $\theta_\nul$ and $R^n_N$ to $\bar{R}^n$ (this is treated formally in Lemmas~\ref{thm:wls_inv} and~\ref{lem:wls_inv_stoch}).
Because $\thetals$ is a consistent estimate of $\theta_\nul$ with an error decaying sufficiently fast, using $\thetals$ in the weighting should not change the asymptotic properties of $\thetawls$. 
Analogously to taking one Gauss-Newton iteration for a maximum likelihood cost function initialized at a strongly consistent estimate, $\thetawls$ is an asymptotically efficient estimate, as will be proven in the next section.

In summary, WNSF consists of the following three steps:
\begin{enumerate}
\item estimate a high-order ARX model with least squares~\eqref{eq:eta_ls};
\item reduce the high-order ARX model to the model of interest with least squares~\eqref{eq:theta_ls};
\item re-estimate the model of interest by weighted least squares~\eqref{eq:theta_wls} using the weighting~\eqref{eq:Wls}.
\end{enumerate}
Two notes can be made about this procedure.
First, the objective of the second step is to obtain a consistent estimate to construct the weighting; hence, the choice of least squares is arbitrary, and weighted least squares with any invertible weighting (e.g., $W_n = R^n_N$) can be used.
Second, although $\hat{\theta}^{\text{WLS}}_N$ is asymptotically efficient, it is possible to continue iterating, which may improve the estimate for finite sample size.

\subsection*{Other Settings}

Despite having been presented for a fully parametrized SISO BJ model, we point out that the method is flexible in parametrization.
For example, it is possible to fix some parameters in $\theta$ if they are known, or to impose linear relations between parameters.
Hence, other common model structures (e.g., OE, ARMA, ARMAX) may also be used, as well as multi-input multi-output (MIMO) versions of such structures.
The requirement is that a relation between the high- and low-order parameters can be written in the form~\eqref{eq:eta-Qtheta=0}.

Moreover, a parametric noise model does not need to be estimated.
In this case, disregard~\eqref{eq:CA-D=0_exp} and consider only~\eqref{eq:FB-LA=0_exp}.
The subsequent steps can then be derived similarly.
This approach is presented in detail and analyzed in \cite{spWNSF}.
In open loop, it provides asymptotically efficient estimates of the dynamic model; in closed loop, the estimates are consistent and with asymptotic covariance corresponding to~\eqref{eq:Mcl}.

\section{Asymptotic Properties}
\label{sec:prop}

We now turn to the asymptotic analysis of WNSF.
Here, we make a distinction between the main algorithm presented here and the aforementioned case without a low-order noise model estimate.
Although apparently simpler because of the smaller dimension of the problem, the case without a noise-model estimate requires additional care in the analysis.
The reason is that the corresponding $T_n(\theta)$ in that case will not be square.
Then, inverting the weighting as in \eqref{eq:Wls} (a relation that will be used for the analysis in this paper) will not be valid, requiring another approach.
Including in this paper under-parametrized noise models is then not possible for space concerns.
Thus, the asymptotic analysis in this paper considers the dynamic and noise models correctly parametrized, in which case the algorithm is consistent and asymptotically efficient.
The case with an under-parametrized noise model (in particular, the limit case where \eqref{eq:CA-D=0_exp} is neglected and no noise-model is estimated) is considered in~\cite{spWNSF}.

Because the ARX model~\eqref{eq:ARXregression} is a truncation of the true system~\eqref{eq:truearx}, its estimate (and the respective covariance) will not be a sufficient statistic for finite order, and some information will be lost in this step.
Then, we need to make sure that, as $N$ grows, the truncation error will be sufficiently small so that, asymptotically, no information is lost.
To keep track of the truncation error in the analysis (see appendices), we let the model order $n$ depend on the sample size $N$---denoting ${n=n(N)}$---according to the following assumption.

\begin{assumption}[ARX-model order]
\label{ass:ARXorder}
It holds that
\begin{itemize}
\item [D1.] $n(N)\to\infty$, as $N\to\infty$;
\item [D2.] $n^{4+\delta}(N)/N \to 0$, for some $\delta>0$, as $N\to\infty$.
\end{itemize}
\end{assumption}
\noindent
Condition D1 implies that, as the sample size $N$ tends to infinity, so does the model order $n$.
Condition D2 establishes a maximum rate at which the model order $n$ is allowed to grow, as we cannot use too high order compared with the number of observations.
A consequence of Condition D2 is that~\cite{ljung&wahlberg92}
\begin{equation}
n^2(N)\log(N)/N \to 0, \text{ as } N\to\infty ,
\label{eq:D3}
\end{equation}
Moreover, defining
$
d(N) := \textstyle{\sum_{k=n(N)+1}^\infty} \abs{a_k^\nul} + \abs{b_k^\nul},
$
we have
\begin{equation}
\sqrt{N}d(N)\to 0, \text{ as } N\to\infty ,
\label{eq:D4}
\end{equation}
as consequence of stability and rational description of the true system in Assumption~\ref{ass:model_truesystem}.
Although~\eqref{eq:D3} and~\eqref{eq:D4} follow from other assumptions, they are stated explicitly as they will be required to show our theoretical results.

To facilitate the statistical analysis, the results in this section consider, instead of~\eqref{eq:eta_ls}, a regularized estimate
\begin{equation}
\hat{\eta}^n_N := \hat{\eta}^{n,\text{reg}}_N = [R^n_\text{reg}(N)]^{\minus 1} r^n_N ,
\label{eq:eta_regls}
\end{equation}
where
\begin{equation}
R^n_\text{reg}(N) =
\left\{
\begin{array}{ll}
R^n_N & \text{if } ||[R^n_N]^{\minus 1}||<2/\delta , \\
R^n_N + \frac{\delta}{2} I_{2n} & \text{otherwise} ,
\end{array}
\right.
\end{equation}
for some small $\delta>0$.
Asymptotically, the first and second order properties of $\hat{\eta}^{n,\text{ls}}_N$ and $\hat{\eta}^n_N$ are identical~\cite{ljung&wahlberg92}.

When we let $n=n(N)$ according to Assumption~\ref{ass:ARXorder}, we use $\hat{\eta}_N := \hat{\eta}^{n(N)}_N$.
We will also denote $\bar{\eta}^{n(N)}$ and $\eta^{n(N)}_\nul$, defined in~\eqref{eq:etahat_conv} and~\eqref{eq:eta0}, respectively.
Concerning the matrices~\eqref{eq:Rbar}, \eqref{eq:Qeta}, \eqref{eq:T0}, \eqref{eq:barW}, and~\eqref{eq:Wls}, for notational simplicity we maintain the subscript $n$ even if $n=n(N)$. 

Some of the technical assumptions used in this paper differ from those used for the asymptotic analysis of PEM~\cite{ljung99}. 
For example, the bound in Assumption~\ref{ass:noise} is stronger than what is required for PEM.
On the other hand, for PEM the parameter vector $\theta$ is required to belong to a compact set, which is not imposed here.
However, such differences in technical assumptions have little relevance in practice.

We have the following result for consistency $\hat{\theta}_N^\text{LS}$. 

\begin{theorem}
\label{thm:consistencyLS}
Let Assumptions~\ref{ass:model_truesystem}, \ref{ass:noise}, \ref{ass:input}, and~\ref{ass:ARXorder} hold, and and $\hat{\theta}_N^\text{LS}$ be defined by~\eqref{eq:theta_ls}.
Then, 
\begin{equation}
\hat{\theta}_N^\text{LS} \to \theta_\nul, \wpone .
\end{equation}
Moreover, we have that
\begin{equation}
||\hat{\theta}^\text{LS}_N-\theta_\nul|| = \mathcal{O} \bigg( \sqrt{n(N)\frac{\log N}{N}}\big(1+d(N)\big) \bigg) .
\label{eq:thetalsdecay}
\end{equation}

\end{theorem}

\begin{proof}
See Appendix~\ref{app:ls_consistency_proof}.
\end{proof}

We have the following result for consistency of $\hat{\theta}_N^\text{WLS}$.

\begin{theorem}
\label{thm:consistencyWLS}
Let Assumptions~\ref{ass:model_truesystem}, \ref{ass:noise}, \ref{ass:input}, and~\ref{ass:ARXorder} hold, and $\hat{\theta}_N^\text{WLS}$ be defined by~\eqref{eq:theta_wls}.
Then, 
\begin{equation}
\hat{\theta}_N^\text{WLS} \to \theta_\nul, \wpone .
\end{equation}
\end{theorem}

\begin{proof}
See Appendix~\ref{app:wls_consistency_proof}.
\end{proof}

We have the following result for asymptotic distribution and covariance of $\hat{\theta}_N^\text{WLS}$.

\begin{theorem}
\label{thm:asycov}
Let Assumptions~\ref{ass:model_truesystem}, \ref{ass:noise}, \ref{ass:input}, and~\ref{ass:ARXorder} hold, and $\hat{\theta}_N^\text{WLS}$ be defined by~\eqref{eq:theta_wls}.
Then, 
\begin{equation}
\sqrt{N}(\hat{\theta}_N^\text{WLS} - \theta_\nul) \sim As\mathcal{N}(0,\sigma_\nul^2 M_{CR}^{\minus 1}) ,
\end{equation}
where
$M_{CR}$ is given by~\eqref{eq:CR}.
\end{theorem}

\begin{proof}
See Appendix~\ref{app:asymp_cov_proof}.
\end{proof}

Theorem~\ref{thm:asycov} implies, comparing with~\eqref{eq:PEMcov}, that WNSF has the same asymptotic properties as PEM. For Gaussian noise, this corresponds to an asymptotically efficient estimate.

\section{Simulation Studies}
\label{sec:sim}

In this section, we perform simulation studies and discuss practical issues.
First, we illustrate the asymptotic properties of the method.
Second, we consider how to choose the order of the non-parametric model.
Third, we exemplify with two difficult scenarios for PEM how WNSF can be advantageous in terms of robustness against convergence to non-global minima and convergence speed.
Fourth, we perform a simulation with random systems to test the robustness of the method compared with other state-of-the-art methods.

Although WNSF and the approach in~\cite{MORSM} are different algorithms, they share the similarities of using high-order models and iterative least squares. 
However, \cite{MORSM} is only applicable in open loop.
Here, to differentiate WNSF as a more general approach that is applicable in open or closed loop without changing the algorithm, we focus on the typically more challenging closed-loop setting, for which many standard methods are not consistent.

\subsection{Illustration of Asymptotic Properties}
\label{subsec:sim:fixedsys}

The first simulation has the purpose of illustrating that the method is asymptotically efficient.
Here, we consider only the case where we estimate a correct noise model (the case where a low-order noise model is not estimated is illustrated in \cite{spWNSF}).
We perform open- and closed-loop simulations, where the closed-loop data are generated by
\begin{equation}
\begin{aligned}
u_t &= \frac{1}{1+K(q)G_\nul(q)} r_t  -  \frac{K(q)H_\nul(q)}{1+K(q)G_\nul(q)} e_t, \\
y_t  &=  \frac{G_\nul(q)}{1+K(q)G_\nul(q)} r_t  + \frac{H_\nul(q)}{1+K(q)G_\nul(q)} e_t, 
\end{aligned}
\label{eq:Kbelow_CL}
\end{equation}
and the open-loop data by
\begin{equation}
u_t = \frac{1}{1+K(q) G_\nul(q)} r_t  , \qquad
y_t = G_\nul(q) u_t  + H_\nul(q) e_t,
\end{equation}
where $\{r_t\}$ and $\{e_t\}$ are independent Gaussian white sequences with unit variance, $K(q)=1$, and
\begin{equation}
G_\nul(q) = \frac{q^{\minus 1} + 0.1q^{\minus 2}}{1 -0.5 q^{\minus 1} + 0.75 q^{\minus 2}} , \quad H_\nul(q) = \frac{1 + 0.7 q^{\minus 1}}{1 - 0.9 q^{\minus 1}} .
\label{eq:sim_G_H}
\end{equation}

We perform 1000 Monte Carlo runs, with sample sizes $N\in\{300,600,1000,3000,6000,10000\}$.
We apply WNSF with an ARX model of order 50 with the open- and closed-loop data.
Performance is evaluated by the mean-squared error of the estimated parameter vector of the dynamic model,
$\text{MSE} = ||\hat{\bar{\theta}}^\text{WLS}_N-\bar{\theta}_\nul||^2$, where $\bar{\theta}$ contains only the elements of $\theta$ contributing to $G(q,\theta)$.
As this simulation has the purpose of illustrating asymptotic properties, initial conditions are zero and assumed known---that is, the sums in~\eqref{eq:RnN} start at $t=1$ instead of $t=n+1$.

The results are presented in Fig.~\ref{fig:asymp_eff}, with the average MSE over 1000 Monte Carlo runs plotted as function of the sample size (closed loop in solid line, open loop in dash-dotted line), where we also plot the corresponding CR bounds (closed loop in dashed line, open loop in dotted line).
The respective CR bounds are attained as the sample size increases.

\begin{figure}
\centering
%
%
\definecolor{mycolor1}{rgb}{0.09804,0.32941,0.65098}%
\definecolor{mycolor3}{rgb}{0.61569,0.06275,0.17647}%
\definecolor{mycolor2}{rgb}{0.38431,0.57255,0.18039}%
\definecolor{mycolor4}{rgb}{0.98039,0.72549,0.09804}%
\definecolor{mycolor5}{rgb}{0.84706,0.32941,0.59216}%
\definecolor{mycolor6}{rgb}{0.39608,0.39608,0.42353}%
\begin{tikzpicture}

\begin{axis}[%
width=.8\figurewidth,
height=.5\figureheight,
at={(0.758in,0.352in)},
scale only axis,
xmode=log,
xmin=300,
xmax=10000,
xminorticks=false,
ymode=log,
ymin=1e-04,
ymax=0.05,
yminorticks=false,
xlabel=$N$,
ylabel=MSE,
ytick={1e-4,1e-3,1e-2},
axis x line=bottom,
axis y line=left,
x label style={at={(axis description cs:.5,0)},anchor=west,font=\footnotesize},
y label style={at={(axis description cs:.06,.5)},anchor=south,font=\footnotesize},
axis background/.style={fill=white},
legend style={legend cell align=left,align=left,draw=white!15!black}
]

\addplot [color=black,solid]
  table[row sep=crcr]{%
300	0.022135320748126\\
600	0.00317826404394755\\
1000	0.00131233978376296\\
3000	0.000359190859135027\\
6000	0.000176358722996808\\
10000	9.93249915139521e-05\\
};


\addplot [color=black,dashdotted]
  table[row sep=crcr]{%
300	0.0242051513897697\\
600	0.00574969826869674\\
1000	0.00269644052860984\\
3000	0.000663632281260718\\
6000	0.000329012724544283\\
10000	0.000184520327609484\\
};

\addplot [color=black,dashed]
  table[row sep=crcr]{%
300	0.00341966666666667\\
600	0.00170983333333333\\
1000	0.0010259\\
3000	0.000341966666666667\\
6000	0.000170983333333333\\
10000	0.00010259\\
};

\addplot [color=black,dotted]
  table[row sep=crcr]{%
300	0.006524\\
600	0.003262\\
1000	0.0019572\\
3000	0.0006524\\
6000	0.0003262\\
10000	0.00019572\\
};

\end{axis}
\end{tikzpicture}%
\caption{Illustration of asymptotic properties: CR bounds in closed loop (dashed) and open loop (dotted), and average MSE for the dynamic-model parameter estimates as function of sample size obtained with WNSF in closed loop (solid) and open loop (dash-dotted).}
\label{fig:asymp_eff}
\end{figure}
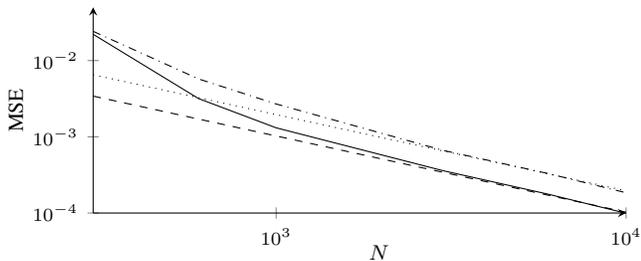

\subsection{Practical Issues}
\label{subsec:sim:practical}

In the previous simulation, an ARX model of order 50 was estimated in the first step.
Although the order of this model should, in theory, tend to infinity at some maximum rate to attain efficiency (Assumption~\ref{ass:ARXorder}), a fixed order was sufficient to illustrate the asymptotic properties of WNSF in this particular scenario.
This suggests that when the number of data samples increases, a non-parametric model of fixed order with sufficiently low bias error may be enough for practical purposes.
However, for fixed sample size, the question remains on how to choose the most appropriate non-parametric model order: a too small $n$ will introduce bias, and a too large $n$ will introduce unnecessary noise in the non-parametric model estimate, which may affect the accuracy of the parametric model estimate.
Some previous knowledge about the speed of the system may help in choosing this order, but the most appropriate value may also depend on sample size and signal-to-noise ratio.
In this paper, we use the PEM cost function~\eqref{eq:J} as criterion to choose $n$: we compute $\hat{\theta}_N^\text{WLS}$ for several $n$, and choose the estimate that minimizes~\eqref{eq:J}.

Also, $\hat{\theta}^\text{WLS}_N$ need not be used as final estimate, as, for finite sample size, performance may improve by iterating.
However, because WNSF does not minimize the cost function~\eqref{eq:J} explicitly, it is not guaranteed that subsequent iterations correspond to a lower cost-function value than previous ones.
Here, we will also use the cost function~\eqref{eq:J} as criterion to choose the best model among the iterations performed.

\subsection{Comparison with PEM}
\label{subsec:sim:PEMcomp}

One of the main limitations of PEM is the non-convex cost function, which may make the method sensitive to the initialization point.
Here, we provide examples illustrating how WNSF may be a more robust method than PEM regarding initialization: in cases where the PEM cost function is highly non-convex, WNSF may require less iterations and be more robust against convergence to non-global minima.

We consider a system where $H_\nul(q)=1$, $K(q)=0.3$, and
\begin{equation}
G_\nul(q) = \frac{1.0 q^{-1}-1.2q^{-2}}{1-2.5q^{-1}+2.4q^{-2}-0.88q^{-3}},
\end{equation}
with data generated according to~\eqref{eq:Kbelow_CL}, where 
\begin{equation}
r_t = \frac{1+0.7q^{-1}}{1-0.9q^{-1}} r_t^w ,
\end{equation}
with $\{e_t\}$ and $\{r_t^w\}$ Gaussian white noise sequences with variances 4 and 0.25, respectively.
The sample size is $N=2000$.
We estimate an OE model with the following algorithms:
\begin{itemize}
\item WNSF with a non-parametric model of order $n=250$;
\item PEM with default MATLAB initialization and Gauss-Newton (GN) algorithm;
\item PEM with default MATLAB initialization (MtL) and Levenberg-Marquardt (LM) algorithm;
\item WNSF with a non-parametric model of order $n=250$, where the weighting matrix, instead of being initialized with $\hat{\theta}_N^\text{LS}$~\eqref{eq:theta_ls}, is initialized with the default MATLAB initialization (MtL);
\item PEM initialized with $\hat{\theta}_N^\text{LS}$ (LS) and the GN algorithm;
\item PEM initialized with $\hat{\theta}_N^\text{LS}$ (LS) and the LM algorithm;
\item PEM initialized at the true parameters (true).
\end{itemize}
All the methods use a maximum of 100 iterations, but stop early upon convergence (default settings for PEM, $10^{\minus 4}$ as tolerance for the normalized relative change in the parameter estimates) and initial conditions are zero.

Performance is evaluated by the FIT of the impulse response of the estimated OE model $G(q,\hat{\theta}_N^\text{WLS})$, given in percent by
\begin{equation}
\text{FIT} = 100\left(1 -\frac{\norm{g_\nul-\hat{g}}}{\norm{g_\nul-\text{mean}(g_\nul)}}\right) ,
\label{eq:fit}
\end{equation}
where $g_\nul$ is a vector with the impulse response parameters of $G_\nul(q)$, and similarly for $\hat{g}$ but for the estimated model.
In~\eqref{eq:fit}, sufficiently long impulse responses are taken to make sure that the truncation of their tails does not affect the FIT.

The average FITs for 100 Monte Carlo runs are shown in Table~\ref{tbl:meanFITfinaliteration}.
For PEM, the results depend on the optimization method and the initialization point: as consequence of the non-convexity of PEM, the algorithms do not always converge to the global optimum.
For PEM implementations, the average FIT is the same as for PEM started at the true parameters only with default MATLAB initialization and LM algorithm.
For WNSF, the average FIT is the same as for PEM started at the true parameters independently of the initialization point used in the weighting matrix, suggesting robustness to different initial weighting matrices.

\begin{table}
  \caption{Comparison with PEM: average FITs with different methods (Meth) and initializations (Init). }
  \label{tbl:meanFITfinaliteration}
  \begin{center}
    \begin{tabular}{ l | *{5}{c} }
      \diagbox{Meth}{Init} & MtL & LS & true \\
      \hline
      WNSF    & 98   & 98  & --  \\
      PEM GN   & 74   & 87  & 98   \\
      PEM LM    & 98   & 85  & 98   \\
      \hline
    \end{tabular}
  \end{center}
  \vspace{-1.3em}
\end{table}

In this simulation, PEM was most robust with the LM algorithm and the default MATLAB initialization, having on average the same accuracy as WNSF.
Then, it is appropriate to compare the performance of these methods by iteration when WNSF also is initialized with the same parameter values.
In Fig.~\ref{fig:iterations}, we plot the average FITs for these methods as function of the maximum number of iterations.
Here, WNSF reaches an average FIT of $98$ after two iterations, while PEM with LM takes 20 iterations to reach the same value.
This suggests that, even if WNSF and some PEM implementation start and converge to the same value, WNSF may do it faster than standard optimization methods for PEM.

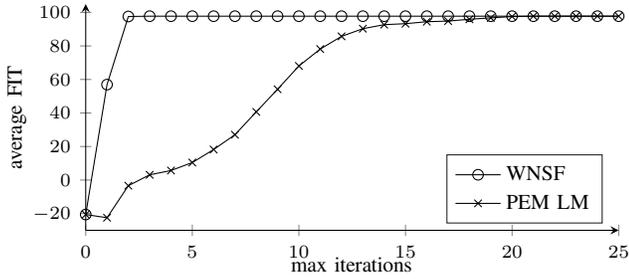
\begin{figure}
\begin{tikzpicture}
  \begin{axis}[%
    width=.8\figurewidth,
    height=.55\figureheight,
    scale only axis,
    xlabel={max iterations},
    axis x line=bottom,
    axis y line=left,
    ymin=-30,
    ymax=105,
    xmax=25,
    xmin=0,
    ylabel={average FIT},
    axis background/.style={fill=white},
    x label style={at={(axis description cs:.5,-.05)},anchor=south,font=\footnotesize},
    y label style={at={(axis description cs:.08,.5)},anchor=south,font=\footnotesize},
    legend style={at={(0.97,0.03)},anchor=south east,legend cell align=left,align=left,draw=white!15!black,font=\footnotesize},
    ytick={-20,0,20,40,60,80,100},
  ]
    \addplot [color=black,solid,mark=o] table[col sep=comma]{Data/iterations_wnsf.csv};
    \addlegendentry{WNSF};
    
    \addplot [color=black,solid,mark=x] table[col sep=comma]{Data/iterations_pem.csv};
    \addlegendentry{PEM LM};

  \end{axis}
\end{tikzpicture}%
\caption{Comparison with PEM: average FIT from 100 Monte Carlo runs function of the maximum number of iterations.}
\label{fig:iterations}
\end{figure}

The robustness of WNSF against convergence to non-global minima compared with different instances of PEM can be even more evident than in Table~\ref{tbl:meanFITfinaliteration}, as WNSF seems to be appropriate for modeling systems with many resonant peaks, for which the PEM cost function can be highly non-linear.
Take the example in Fig.~\ref{fig:peakssys}, based on 100 Monte Carlo runs for a system with
\begin{equation}
\begin{aligned}
&L_\nul(q) = q^{-1}-3.4q^{-2}+4.8q^{-3}-3.3q^{-4}+0.96q^{-5}, \\
&\begin{multlined}
F_\nul(q) = 1-5.4q^{-1}+13.5q^{-2}-20.1q^{-3}+19.5q^{-4}\\-12.1q^{-5}+4.5q^{-6},
\end{multlined}
\end{aligned}
\end{equation}
and data generated according to~\eqref{eq:Kbelow_CL} with $K(q)=-0.05$,
\begin{equation}
r_t = \frac{0.05}{1-0.99q^{-1}}r_t^w,
\end{equation}
where $\{r_t^w\}$ and $\{e_t\}$ are Gaussian white sequences with unit variance.
Here, initial conditions are not assumed zero: PEM estimates initial conditions by backcasting and WNSF uses the approach in~\cite{galrinho15}.
In this scenario, PEM with the LM algorithm and default initialization fails in most runs to find the global optimum. 
Subspace methods, often used to avoid the non-convexity of PEM, may not help in this scenario: SSARX~\cite{jansson03}, a subspace method that is consistent in closed loop, provides an average FIT around $20\%$ (default MATLAB implementation).
Here, WNSF with $n$ between 100 and 600 spaced with intervals of 50 performs similarly to PEM initialized at the true parameters, accurately capturing the resonance peaks of the system.

\subsection{Random Systems}
\label{subsec:sim:randsys}

In order to test the robustness of the method, we now perform a simulation with random systems.
Also, closed-loop data often introduces additional difficulties: for example, many standard methods are not consistent.
Thus, we perform a simulation with these settings and compare the performance of WNSF with other methods available in the Mathworks MATLAB System Identification Toolbox.
For a fair comparison, we only use methods that are consistent in closed loop and only use input and output data.
From the subspace class, we use SSARX, as this method is competitive with other subspace algorithms such as CVA~\cite{larimore83,petsch96} and N4SID~\cite{vanodem94}, while it is consistent in closed loop~\cite{jansson03}.
IV methods are not considered, as they in closed loop require the reference signal to construct the instruments.

For the simulation, we use 100 systems with structure
\begin{equation}
G_\nul(q) = \frac{l_1^\nul q^{\minus 1} + \cdots + l_4^\nul q^{\minus 4}}{1+f_6^\nul q^{\minus 6} + \cdots + f_m^\nul q^{\minus 6}} . 
\end{equation}
As we have observed, PEM may have difficulties with slow resonant systems: therefore, it is for this class of systems that WNSF may be most beneficial.
With this purpose, we generate the polynomial coefficients in the following way.
The poles are located in an annulus with the radius uniformly distributed between $0.88$ and $0.98$, and the phase uniformly distributed between 0 and $90^\circ$ (and respective complex conjugates).
One pair of zeros is generated in the same way, and a third real zero is uniformly distributed between $-1.2$ and $1.2$ (this allows for non-minimum-phase systems).
The noise models have structure
\begin{equation}
H_\nul(q) = \frac{1 + c_1^\nul q^{\minus 1} + c_2^\nul q^{\minus 2}}{1+d_1^\nul q^{\minus 1} + d_2^\nul q^{\minus 2}}, 
\end{equation}
with the poles and zeros having uniformly distributed radius between 0 and 0.95, and uniformly distributed phase between 0 and 180$^\circ$ (and respective complex conjugates).

The data are generated in closed loop by 
\begin{equation}
\begin{aligned}
u_t &= \frac{K(q)}{1+K(q)G_\nul(q)} r_t  -  \frac{K(q)H_\nul(q)}{1+K(q)G_\nul(q)} e_t , \\
y_t  &=  \frac{K(q)G_\nul(q)}{1+K(q)G_\nul(q)} r_t  + \frac{H_\nul(q)}{1+K(q)G_\nul(q)} e_t, \\
\end{aligned}
\end{equation}
where
\begin{equation}
r_t = \frac{1-1.273q^{-1}+0.81q^{-2}}{1-1.559q^{-1}+0.81q^{-2}} r_t^w
\end{equation}
with $\{r_t^w\}$ a Gaussian white-noise sequence with unit variance, $\{e_t\}$ a Gaussian white-noise sequence with the variance chosen such that the signal-to-noise ratio (SNR) is
\begin{equation}
\text{SNR} = \frac{ \sum_{t=1}^N \left[ \frac{K(q)G_\nul(q)}{1+K(q)G_\nul(q)} r_t \right]^2 }{ \sum_{t=1}^N \left[ H_\nul(q) e_t \right]^2 } = 2,
\end{equation}
and the controller $K(q)$ is obtained using a Youla-parametrization to have an integrator and a closed-loop transfer function that has the same poles as the open loop except that the radius of the slowest open-loop pole pair is reduced by 80\%.
The sample size is $N=2000$ and we perform 100 Monte Carlo runs (one for each system; different noise realizations).

We compare the following methods:
\begin{itemize}
\item PEM initialized at the true parameters (PEMt);
\item PEM with default MATLAB initialization (PEMd);
\item SSARX with the default MATLAB options;
\item WNSF using the approach in Section~\ref{subsec:sim:practical} to choose $n$ from the grid $\{50,100,150,200,250,300\}$.
\item PEM initialized with WNSF (PEMw).
\end{itemize}
All methods estimate a fully parametrized noise model.
We use the MATLAB2016b System Identification Toolbox implementation of SSARX and PEM.
For PEM, the optimization algorithm is LM.
For SSARX, the horizons are chosen automatically by MATLAB, based on the Akaike Information Criterion.
WNSF and PEM use a maximum of 100 iterations, but stop earlier upon convergence (default settings for PEM, $10^{\minus 4}$ as tolerance for the normalized relative change in the parameter estimates.
PEM estimates initial conditions by backcasting and WNSF truncates them (\cite{galrinho15} does not apply to BJ models).

The FITs obtained in this simulation are presented in Fig.~\ref{fig:randomsys}.
In this scenario, PEM with default MATLAB initialization (PEMd) often fails to find a point close to the global optimum, which can be concluded by comparison with PEM initialized at the true parameters (PEMt).
Also, SSARX is not an alternative for achieving better performance.
WNSF can be an appropriate alternative, failing only once to provide an acceptable estimate, and having otherwise a performance close to the practically infeasible PEMt.
The estimate obtained with WNSF may be used to initialize PEM.
This provides a small improvement only, suggesting that the estimates obtained with WNSF are already close to a (local) minimum of the PEM cost function.

\begin{figure}
\pgfplotsset{compat=1.14}
\begin{tikzpicture}
  \begin{axis}[
    boxplot/draw direction=y,
    width=\figurewidth,
    height=.85\figureheight,
    ymin  = -30,
    x axis line style={opacity=0},
    axis x line*=bottom,
    axis y line=left,
    enlarge y limits,
    xtick style={draw=none},
    xticklabels={,,PEMd,SSARX,WNSF,PEMw,PEMt},
    ylabel=FIT,
    label style={font=\small},
    tick label style={font=\small},
    ytick={-25,0,25,50,75,100},
  ]
    \addplot [boxplot,mark=x] table[y index=0]{Data/randomsys_pem.csv};
    \addplot [boxplot,mark=x] table[y index=0]{Data/randomsys_ssarx.csv};
    \node at (axis cs:2,-26) [anchor=north] {\scriptsize{$\downarrow 3$}};
    \addplot [boxplot,mark=x] table[y index=0]{Data/randomsys_wnsf.csv};
    \addplot [boxplot,mark=x] table[y index=0]{Data/randomsys_pem_wnsf.csv};
    \addplot [boxplot,mark=x] table[y index=0]{Data/randomsys_pem_true.csv};
  \end{axis}
\end{tikzpicture}
\caption{Random systems: FITs from 100 Monte Carlo runs.}
\label{fig:randomsys}
\end{figure}
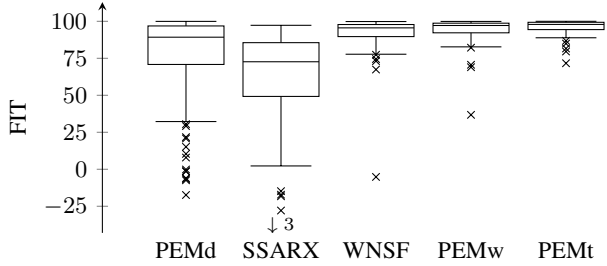

\section{Conclusion}

Methods for parameter estimation based on an intermediate unstructured model have a long history in system identification (e.g., \cite{wahlberg89,zhu_book,jansson03,bjsm}).
Here, we believe to have taken a significant step further in this class of methods, with a method that is flexible in parametrization and provides consistent and asymptotically efficient estimates in open and closed loop without using a non-convex optimization or iterations.

In this paper, we provided a theoretical and experimental analysis of this method, named weighted null-space fitting (WNSF).
Theoretically, we showed that the method is consistent and asymptotically efficient for stable Box-Jenkins systems.
Experimentally, we performed Monte Carlo simulations, comparing PEM, subspace, and WNSF under settings where PEM typically performs poorly.
The simulations suggest that WNSF is competitive with these methods, being a viable alternative to PEM or to provide initialization points for PEM.

Although WNSF was here presented for SISO BJ models, it was also pointed out that the flexibility in parametrization allows for a wider range of structures to be used, as well as for incorporating structural information (e.g., fixing specified parameters).
Moreover, based on the analysis in~\cite{spWNSF}, WNSF does not require a parametric noise model to achieve asymptotic efficiency in open loop and consistency in closed loop.

An extension that was not covered in this paper is the MIMO case, where subspace or IV methods are typically used~\cite{stoicajansson_mimo}, as PEM often has difficulty with estimation of such systems.
Based on the theoretical foundation provided in this contribution, this important extension is already in preparation. 
Future work includes also extensions to dynamic networks and non-linear model structures.

\appendices

\section{Auxiliary Results}

In this appendix, we present some results that will be applied in the remainder of the paper.

\begin{prop}
\label{lemma:etabar-eta0}
Let Assumptions~\ref{ass:model_truesystem}, \ref{ass:noise}, and \ref{ass:input} hold.
Also, let $\bar{\eta}^n$ be defined by \eqref{eq:etahat_conv} and $\eta^n_\nul$ by \eqref{eq:eta0}. 
Then,
\begin{equation}
\textstyle
\norm{\bar{\eta}^n-\eta^n_\nul} \leq C \sum_{k=n+1}^\infty \abs{a^\nul_k} + \abs{b^\nul_k} \to 0, \text{ as } n\to\infty.
\end{equation}
\end{prop}

\begin{proof}
The result follows from \cite[Lemma 5.1]{ljung&wahlberg92} and~\eqref{eq:D4}.
\end{proof}

\begin{prop}
\label{thm:hateta-bareta}
Let Assumptions~\ref{ass:model_truesystem}, \ref{ass:noise}, \ref{ass:input}, and~\ref{ass:ARXorder} hold.
Also, let $\hat{\eta}_N:=\hat{\eta}^{n(N)}_N$ be defined by~\eqref{eq:eta_regls} and $\bar{\eta}^{n(N)}$ by \eqref{eq:etahat_conv}.
Then,
\begin{equation}
\big|\hspace{-1pt}\big|\hat{\eta}_N-\bar{\eta}^{n(N)}\big|\hspace{-1pt}\big| =
\mathcal{O}\bigg( \sqrt{\frac{n(N)\log N}{N}} [1+d(N)] \bigg) ,
\label{eq:hateta-bareta_rate}
\end{equation}
and
$
\big|\hspace{-1pt}\big|\hat{\eta}_N-\bar{\eta}^{n(N)}\big|\hspace{-1pt}\big| \to 0, \wpone.
$
\end{prop}

\begin{proof}
For the first part, see \cite[Theorem 5.1]{ljung&wahlberg92}.
The second part follows from~\eqref{eq:D3} and~\eqref{eq:D4}.
\end{proof}

\begin{prop}
\label{lemma:RnN-barRn}
Let Assumptions~\ref{ass:model_truesystem}, \ref{ass:noise}, \ref{ass:input}, and~\ref{ass:ARXorder} hold.
Then,
\begin{equation}
\big|\hspace{-1pt}\big|R^{n(N)}_N-\bar{R}^{n(N)}\big|\hspace{-1pt}\big| = \mathcal{O} \left( 2 n(N) \sqrt{\frac{\log N}{N}} + C \frac{n^2(N)}{N}\right).
\end{equation}
\end{prop}

\begin{proof}
See~\cite[Lemma 4.1]{ljung&wahlberg92}.
\end{proof}

\begin{prop}
\label{thm:Upsilon}
Let Assumptions~\ref{ass:model_truesystem}, \ref{ass:noise}, \ref{ass:input}, and~\ref{ass:ARXorder} hold.
Also, let $\Upsilon^n$ be an $m\times 2n$ deterministic matrix, with $m$ fixed.
Then, 
\begin{equation}
\sqrt{N} \Upsilon^n (\hat{\eta}_N-\bar{\eta}^{n(N)}) \sim As\mathcal{N} (0,P) ,
\end{equation}
where
$
P = \sigma_\nul^2 \lim\limits_{n\to\infty} \Upsilon^n [\bar{R}^n]^{\minus 1} (\Upsilon^n)^\top ,
$
if the limit exists.
\end{prop}

\begin{proof}
See~\cite[Theorem 7.3]{ljung&wahlberg92}.
\end{proof}

\begin{prop}
\label{prop:Xdiff->0}
Consider the product $\prod_{i=1}^p X^{(i)}_N$, where $p$ is finite and $X^{(i)}_N$ are stochastic matrices of appropriate dimensions (possibly a function of $N$) such that
\begin{equation}
||X^{(i)}_N - \bar{X}^{(i)}|| \to 0, \wpone
\end{equation}
and $||\bar{X}^{(i)}||<C_i$.
Then, we have that
\begin{equation}
\textstyle
\big|\hspace{-1pt}\big|\prod_{i=1}^p X^{(i)}_N - \prod_{i=1}^p \bar X^{(i)}\big|\hspace{-1pt}\big| \to 0, \wpone.
\label{eq:prodX-prodXbar->0}
\end{equation}
\end{prop}

\begin{proof}
We show this by induction.
First, let $p=2$ and define $\Delta_N^{(i)}:=X_N^{(i)}-\bar{X}_N^{(i)}$. Then, we can write
\begin{equation}
X^{(1)}_N X^{(2)}_N - \bar{X}^{(1)}\bar{X}^{(2)} =
\Delta^{(1)}_N \bar X^{(2)} + \bar{X}^{(1)} \Delta^{(2)}_N + \Delta^{(1)}_N \Delta^{(2)}_N,
\label{eq:rewriteX1X2-X1X2bar}
\end{equation}
which yields, using the assumptions,
\begin{equation}
\begin{multlined}
||X^{(1)}_N X^{(2)}_N \!-\! \bar{X}^{(1)}\bar{X}^{(2)}|| \!\leq\! ||\Delta^{(1)}_N|| \, ||\bar X^{(2)}|| \\ \!\!\!\!+\! ||\bar X^{(1)}|| \, || \Delta^{(2)}_N|| \!+\! ||\Delta^{(1)}_N|| \, ||\Delta^{(2)}_N|| \to 0 \wpone.
\end{multlined}
\label{eq:X1X2pN-X1X2pbar}
\end{equation}
Second, we consider an arbitrary $p$, and assume that
\begin{equation}
\textstyle
\big|\hspace{-1pt}\big|\prod_{i=1}^{p-1} X^{(i)}_N - \prod_{i=1}^{p-1} \bar X^{(i)}\big|\hspace{-1pt}\big| \to 0, \wpone.
\label{eq:prodX-prodXbar->0_p-1}
\end{equation}
Then, using a similar procedure as~\eqref{eq:rewriteX1X2-X1X2bar}, we have
\begin{equation}
\begin{multlined}
\textstyle
\big|\hspace{-1pt}\big|\prod_{i=1}^{p} X^{(i)}_N - \prod_{i=1}^{p} \bar X^{(i)}\big|\hspace{-1pt}\big| \leq ||\Delta^{(p)}_N|| \,  || \prod_{i=1}^{p-1} \bar X^{(i)}|| \\ \textstyle + ||\bar X^{(p)}|| \; || \prod_{i=1}^{p-1} \Delta^{(i)}_N||+ ||\Delta^{(p)}_N|| \; ||\prod_{i=1}^{p-1}\Delta^{(i)}_N||,
\end{multlined}
\label{eq:X1X2pN-X1X2pbar}
\end{equation}
which, in turn, is bounded by
\begin{equation}
\begin{multlined}
\textstyle
||\Delta^{(p)}_N|| \, \prod_{i=1}^{p-1} ||  \bar X^{(i)}||  + ||\bar X^{(p)}|| \; \prod_{i=1}^{p-1} ||  \Delta^{(i)}_N|| \\ 
\textstyle+ ||\Delta^{(p)}_N|| \; \prod_{i=1}^{p-1} ||\Delta^{(i)}_N||
\to 0 \wpone,
\end{multlined}
\end{equation}
where the convergence follows by assumption.
Then, \eqref{eq:prodX-prodXbar->0} is verified when assuming~\eqref{eq:prodX-prodXbar->0_p-1}, which considering also~\eqref{eq:X1X2pN-X1X2pbar} and an induction argument, concludes the proof.
\end{proof}

%

\section{Consistency of Step 2}

The main purpose of this appendix is to prove Theorem~\ref{thm:consistencyLS}.
However, before we do so, we introduce some results regarding the norm of some vectors and matrices.

\begin{itemize}[leftmargin=*]
\item $\big|\hspace{-1pt}\big|\hat{\eta}_N-\eta^{n(N)}_\nul\big|\hspace{-1pt}\big|$ \textit{tends to zero, as $N$ tends to infinity, w.p.1}
\end{itemize}
Consider the estimated parameter vector ${\hat{\eta}_N:=\hat{\eta}^{n(N)}_N}$~\eqref{eq:eta_regls}, and the truncated true parameter vector $\eta_\nul^{n(N)}$~\eqref{eq:eta0}.
Using the triangular inequality, we have
\begin{equation}
\big|\hspace{-1pt}\big|\hat{\eta}_N-\eta^{n(N)}_\nul\big|\hspace{-1pt}\big| \leq \big|\hspace{-1pt}\big|\hat{\eta}_N-\bar{\eta}^{n(N)}\big|\hspace{-1pt}\big| + \big|\hspace{-1pt}\big|\bar{\eta}^{n(N)}-\eta^{n(N)}_\nul\big|\hspace{-1pt}\big| ,
\label{eq:hateta-eta0n}
\end{equation}
where $\bar{\eta}^n$ is defined by~\eqref{eq:etahat_conv}.
Then, from Proposition~\ref{lemma:etabar-eta0}, the second term on the right side of~\eqref{eq:hateta-eta0n} tends to zero as ${n(N)\to\infty}$. 
From Proposition~\ref{thm:hateta-bareta}, the first term on the right side of~\eqref{eq:hateta-eta0n} tends to zero, as ${N\to\infty}$, w.p.1. 
Thus,
\begin{equation}
\big|\hspace{-1pt}\big|\hat{\eta}_N-\eta^{n(N)}_\nul\big|\hspace{-1pt}\big| \to 0 , \wpone .
\label{eq:hateta-eta0}
\end{equation}

\begin{itemize}[leftmargin=*]
\item $\big|\!\big|Q_n(\hat{\eta}_N)-Q_n(\eta^{n(N)}_\nul)\big|\!\big|$ \textit{tends to zero, as $N$ to infinity, w.p.1}
\end{itemize}
Consider $Q_n(\eta_\nul^{n(N)})$, given by~\eqref{eq:Qeta} evaluated at the truncated true parameter vector $\eta_\nul^{n(N)}$, and the matrix $Q_n(\hat{\eta}_N)$, given by~\eqref{eq:Qeta} evaluated at the estimated parameters $\hat{\eta}_N$.
We have
\begin{equation}
\begin{aligned}
&\big|\hspace{-1pt}\big|Q_n(\hat{\eta}_N)-Q_n(\eta^{n(N)}_\nul)\big|\hspace{-1pt}\big| \leq \big|\hspace{-1pt}\big|Q^c_n(\hat{\eta}_N)-Q^c_n(\eta^{n(N)}_\nul)\big|\hspace{-1pt}\big| \\
&\quad+ \big|\hspace{-1pt}\big|Q^l_n(\hat{\eta}_N)-Q^l_n(\eta^{n(N)}_\nul)\big|\hspace{-1pt}\big| + \big|\hspace{-1pt}\big|Q^f_n(\hat{\eta}_N)-Q^f_n(\eta^{n(N)}_\nul)\big|\hspace{-1pt}\big|\\ 
&\leq C \big|\hspace{-1pt}\big|\hat{\eta}_N-\eta^{n(N)}_\nul\big|\hspace{-1pt}\big| .
\end{aligned}
\label{eq:Qhateta-Qeta0_proof}
\end{equation}
Then, using~\eqref{eq:hateta-eta0}, we conclude that
\begin{equation}
\big|\hspace{-1pt}\big|Q_n(\hat{\eta}_N)-Q_n(\eta^{n(N)}_\nul)\big|\hspace{-1pt}\big| \to 0 , \wpone .
\label{eq:Qhateta-Qeta0}
\end{equation}

\begin{itemize}[leftmargin=*]
\item $\big|\hspace{-1pt}\big|Q_n(\eta_\nul^n)\big|\hspace{-1pt}\big|$ \textit{is bounded for all} $n$
\end{itemize}
We have that 
\begin{equation}
\begin{aligned}
\norm{Q_n(\eta_\nul^n)} &\leq 
\norm{Q^c_n(\eta_\nul^n)} \!+\! \norm{Q^l_n(\eta_\nul^n)} \!+\! \norm{Q^f_n(\eta_\nul^n)} \!+\! \norm{Q^d_n} \\
&\leq C \norm{\eta_\nul^n} + 1 \leq C \norm{\eta_\nul} + 1 , \forall n \\
\end{aligned}
\label{eq:bound_Qeta0}
\end{equation}
which is bounded, by stability of the true system.

\begin{itemize}[leftmargin=*]
\item $\norm{Q_n(\hat{\eta}_N)}$ \textit{is bounded for large $N$, w.p.1}
\end{itemize}
Using the triangular inequality, we have
\begin{equation}
\norm{Q_n(\hat{\eta}_N)} \leq \big|\hspace{-1pt}\big| Q_n(\hat{\eta}_N)-Q_n(\eta^{n(N)}_\nul)\big|\hspace{-1pt}\big| + \big|\hspace{-1pt}\big| Q_n(\eta_\nul^{n(N)}) \big|\hspace{-1pt}\big| .
\label{eq:Qhat_triang}
\end{equation}
Using now~\eqref{eq:bound_Qeta0} and~\eqref{eq:Qhateta-Qeta0_proof}, the first term on the right side of~\eqref{eq:Qhat_triang} can be made arbitrarily small as $N$ increases, while the second term is bounded for all $n(N)$.
Then, there exists $\bar{N}$ such that
\begin{equation}
\norm{Q_n(\hat{\eta}_N)} \leq C , \quad \forall N>\bar{N} .
\label{eq:bound_Qhateta}
\end{equation}

\begin{itemize}[leftmargin=*]
\item $\norm{T_n(\theta_\nul)}$ \textit{is bounded for all $n$}
\end{itemize}
Consider the matrix $T_n(\theta_\nul)$, given by~\eqref{eq:T0}.
First, we introduce the following result.
Let $X(q)=\sum_{k=0}^\infty x_k q^{\minus k}$ and define 
\begin{equation}
\mathbb{T}[X(q)] :=
\begin{bmatrix}
x_0    & 0      & 0      & \cdots \\
x_1    & x_0    & 0      & \ddots \\
x_2    & x_1    & x_0    & \ddots \\
\vdots & \ddots & \ddots & \ddots
\end{bmatrix} .
\label{eq:mathbbT}
\end{equation}
If $\sqrt{\sum_{k=0}^\infty |x_k|^2}<C_1$, we have that~\cite{peller}
\begin{equation}
\norm{\mathbb{T}[X(q)]} \leq C .
\label{eq:Tinfbounded}
\end{equation}
When $X(q)$ can be written as a rational transfer function, \eqref{eq:Tinfbounded} follows from $X(q)$ having all poles strictly inside the unit circle, as, in this case, the sum of squares of its impulse response coefficients is bounded.

We observe that the blocks of $T_n(\theta_\nul)$ satisfy that $T^f_n(\theta_\nul)$, $T^c_n(\theta_\nul)$, and $T^l_n(\theta_\nul)$ are sub-matrices of $\mathbb{T}[F_\nul(q)]$, $\mathbb{T}[C_\nul(q)]$, and $\mathbb{T}[L_\nul(q)]$, respectively.
Then, we have that
\begin{equation}
\begin{aligned}
\norm{T_n(\theta_\nul)} &\leq \norm{T^f_n(\theta_\nul)} + \norm{T^c_n(\theta_\nul)} + \norm{T^l_n(\theta_\nul)} \\
&\leq \norm{\mathbb{T}[F_\nul(q)]} + \norm{\mathbb{T}[C_\nul(q)]} + \norm{\mathbb{T}[L_\nul(q)]} \\
&\leq C \;\; \forall n ,
\end{aligned}
\label{eq:T0_bounded}
\end{equation}
where the last inequality follows from~\eqref{eq:Tinfbounded} and from $F(q)$, $C(q)$, and $L(q)$ being finite order polynomials.

The following lemma is useful for invertibility of the least-squares problem~\eqref{eq:theta_ls}.

\begin{lemma}
\label{thm:ls_inv}
Let Assumption~\ref{ass:model_truesystem} hold and
\begin{equation}
M(\eta_\nul) :=\lim\limits_{n\to\infty} Q_n^\top(\eta^n_\nul)Q_n(\eta^n_\nul), 
\label{eq:Meta0def}
\end{equation}
where $Q_n(\eta^n_\nul)$ is given by~\eqref{eq:Qeta} evaluated at $\eta^n_\nul$, defined in~\eqref{eq:eta0}.
Then, $M(\eta_\nul)$ is invertible.
\end{lemma}

\begin{proof}
Firts, we observe that the limit in~\eqref{eq:Meta0def} is well defined, because the entries of $M(\eta^n_\nul):=Q^\top(\eta^n_\nul)Q(\eta^n_\nul)$ are either zero or sums with form
\begin{equation}
\textstyle
\sum_{k=1}^n a_k^\nul a_{k+p}^\nul , \;\;
\sum_{k=1}^n a_k^\nul b_{k+p}^\nul , \;\;
\sum_{k=1}^n b_k^\nul b_{k+p}^\nul ,
\end{equation}
for some finite integers $p$, and the coefficients $a_k^\nul$ and $b_k^\nul$ are stable sequences.
Thus, these sums converge as $n\to\infty$.
For simplicity of notation, let $Q_\infty(\eta_\nul):=\lim_{n\to\infty} Q_n(\eta^n_\nul)$; that is, $Q_\infty(\eta_\nul)$ is block Toeplitz according to~\eqref{eq:Qeta}, with each block having an infinite number of rows and given by
\begin{align}
&Q_\infty^c(\eta_\nul) \!=\! \mathcal{T}_{\infty,m_c}(A(q,\eta_\nul)), \;\; 
Q_\infty^l(\eta_\nul) \!=\! \mathcal{T}_{\infty,m_l}(A(q,\eta_\nul)), \\
&Q_\infty^f(\eta_\nul) \!=\! \mathcal{T}_{\infty,m_f}(B(q,\eta_\nul)), \;\;
Q_\infty^d \!=\! 
\begin{bmatrix}
I_{m_d,m_d} \\ 0_{\infty,m_d}
\end{bmatrix} \label{eq:Qtoepinf}.
\end{align}
We can then write
$
M(\eta_\nul) = Q^\top_\infty(\eta_\nul) Q_\infty(\eta_\nul).
$
From this factorization, we observe that $M(\eta_\nul)$ is singular if and only if $Q_\infty(\eta_\nul)$ has a non-trivial right null-space. 
Moreover, the block anti-diagonal structure of $Q_\infty(\eta_\nul)$ implies that $Q_\infty$ has full column rank if and only if both matrices $[-Q_\infty^f(\eta_\nul) \; Q_\infty^l(\eta_\nul)]$ and $[-Q_\infty^c(\eta_\nul) \; Q_\infty^d(\eta_\nul)]$ have full column rank.
We proceed by contradiction. Suppose that
\begin{equation}
\begin{bmatrix}
-Q_\infty^f(\eta_\nul) & Q_\infty^l(\eta_\nul)
\end{bmatrix}
\begin{bmatrix}
\alpha \\ \beta
\end{bmatrix}
= 
-Q_\infty^f(\eta_\nul)\alpha + Q_\infty^l(\eta_\nul)\beta = 0,
\label{eq:-Qfa+Qfb=0}
\end{equation}
where $\alpha$ and $\beta$ are some vectors $\alpha = [\alpha_0 \; \dots \; \alpha_{m_{f-1}}]^\top$ and $\beta = [\beta_0 \; \dots \; \beta_{m_{l-1}}]^\top$.
Then, \eqref{eq:-Qfa+Qfb=0} implies
\begin{equation}
B(q,\eta_\nul) \alpha(q) = A(q,\eta_\nul) \beta(q) \Leftrightarrow L(q,\theta_\nul) \alpha(q) = F(q,\theta_\nul) \beta(q) ,
\label{eq:La=Fb}
\end{equation}
where $\alpha(q)=\sum_{k=0}^{m_f-1} \alpha_k q^{-k}$ and $\beta(q)=\sum_{k=0}^{m_l-1} \beta_k q^{-k}$.
Because $L(q,\theta_\nul)$ and $F(q,\theta_\nul)$ are co-prime by Assumption~\ref{ass:model_truesystem} and polynomials of order $m_l-1$ and $m_f$, and $\alpha(q)$ and $\beta(q)$ are polynomials of orders at most $m_f-1$ and $m_l-1$, \eqref{eq:La=Fb} can only be satisfied if $\alpha(q)\equiv 0 \equiv\beta(q)$.
Hence, $[-Q_\infty^f(\eta_\nul) \; Q_\infty^l(\eta_\nul)]$ has full column rank.

Analogously for $[-Q_\infty^c(\eta_\nul) \; Q_\infty^d(\eta_\nul)]$, this matrix has full column rank if and only if
$
C(q,\theta_\nul) \alpha(q) = D(q,\theta_\nul) \beta(q)
$
is satisfied only for $\alpha(q)\equiv 0 \equiv\beta(q)$, where here we have $\alpha(q)=\sum_{k=0}^{m_c-1} \alpha_k q^{-k}$ and $\beta(q)=\sum_{k=0}^{m_d-1} \beta_k q^{-k}$.
This is the case, as $C(q,\theta_\nul)$ and $D(q,\theta_\nul)$ are co-prime and polynomials of higher order than $\alpha(q)$ and $\beta(q)$.
Hence, $[-Q_\infty^f(\eta_\nul) \; Q_\infty^l(\eta_\nul)]$ and $[-Q_\infty^c(\eta_\nul) \; Q_\infty^d(\eta_\nul)]$ are full column rank, implying that $Q_\infty(\eta_\nul)$ has a trivial right null-space and $M(\eta_\nul)$ is invertible.
\end{proof}

Finally, we have the necessary results to prove Theorem~\ref{thm:consistencyLS}.

\subsubsection*{Proof of Theorem \ref{thm:consistencyLS}}
\label{app:ls_consistency_proof}

We start by using~\eqref{eq:theta_ls} to write 
\begin{equation}
\begin{aligned}
\thetals\!\!-\!\theta_\nul \! &= \! \left[ Q_n^\top\!(\hat{\eta}_N\!) Q_n(\hat{\eta}_N\!) \right]^{\minus 1} \!\!Q_n^\top\!(\hat{\eta}_N\!) \hat{\eta}_N - \theta_\nul \\
&= \left[ Q_n^\top\!(\hat{\eta}_N\!) Q_n(\hat{\eta}_N\!) \right]^{\minus 1} \!\!Q_n^\top\!(\hat{\eta}_N\!) \left[\hat{\eta}_N-Q_n(\hat{\eta}_N\!) \theta_\nul\right] \\
&= \left[ Q_n^\top\!(\hat{\eta}_N\!) Q_n(\hat{\eta}_N\!) \right]^{\minus 1} \!\! Q_n^\top\!(\hat{\eta}_N\!) T_n(\theta_\nul)\! [\hat{\eta}_N\!-\!\eta^{n(N)}_\nul]
\end{aligned}
\label{eq:thetals-theta0}
\end{equation}
where the last equality follows from~\eqref{eq:Tetadiff}.
If $n$ were fixed, consistency would follow if $\hat{\eta}_N-\eta^n_\nul$  would approach zero as $N\to\infty$, provided the inverse of $Q_n^\top(\hat{\eta}_N) Q_n(\hat{\eta}_N)$ existed for sufficiently large $N$.
However, $n=n(N)$ increases according to Assumption~\ref{ass:ARXorder}.
This implies that the dimensions of the vectors $\hat{\eta}_N$ and $\eta^{n(N)}_\nul$, and of the matrices $Q_n(\hat{\eta}_N)$ (number of rows) and $T_n(\theta_\nul)$ (number of rows and columns), become arbitrarily large.
Therefore, extra requirements are necessary.
In particular, we use~\eqref{eq:thetals-theta0} to write
\begin{align}
&\big|\hspace{-1pt}\big|\thetals-\theta_\nul\big|\hspace{-1pt}\big| = \big|\hspace{-1pt}\big|M^{\minus 1}(\hat{\eta}_N) Q_n^\top(\hat{\eta}_N) T_n(\theta_\nul) (\hat{\eta}_N-\eta^{n(N)}_\nul)\big|\hspace{-1pt}\big| \\
&\leq \big|\hspace{-1pt}\big|M^{\minus 1}(\hat{\eta}_N)\big|\hspace{-1pt}\big| \, \big|\hspace{-1pt}\big|Q_n(\hat{\eta}_N)\big|\hspace{-1pt}\big| \, \big|\hspace{-1pt}\big|T_n(\theta_\nul)\big|\hspace{-1pt}\big| \, \big|\hspace{-1pt}\big|\hat{\eta}_N-\eta^{n(N)}_\nul\big|\hspace{-1pt}\big| ,
\label{eq:norm_thetals-theta0}
\end{align}
where $M(\hat{\eta}_N):=Q_n^\top(\hat{\eta}_N) Q_n(\hat{\eta}_N)$. 
Consistency is achieved if the last factor on the right side of the inequality in~\eqref{eq:norm_thetals-theta0} approaches zero, as $N\to\infty$, w.p.1, and the remaining factors are bounded for sufficiently large $N$, w.p.1.
This can be shown using~\eqref{eq:bound_Qhateta}, \eqref{eq:T0_bounded}, and~\eqref{eq:hateta-eta0}, but we need additionally that $M(\hat\eta_N)$ is invertible for sufficiently large $N$, w.p.1.

With this purpose, we write 
\begin{equation}
\begin{multlined}
||M(\hat{\eta}_N)- M(\eta^{n(N)}_\nul)|| \\
= ||Q_n^\top(\hat{\eta}_N)Q_n(\hat{\eta}_N)-Q_n^\top(\eta^{n(N)}_\nul)Q_n(\eta^{n(N)}_\nul)|| 
\end{multlined}
\label{eq:Mhat-M0_ls}
\end{equation}
Using~\eqref{eq:Qhateta-Qeta0_proof}, \eqref{eq:bound_Qeta0}, \eqref{eq:bound_Qhateta}, and Proposition~\ref{prop:Xdiff->0}, and because ${M(\eta^n_\nul)\to M(\eta_\nul)}$ as ${n\to\infty}$, we have that
\begin{equation}
M(\hat{\eta}_N) \to M(\eta_\nul), \wpone.
\label{eq:Mhateta->M0}
\end{equation}
As $M(\eta_\nul)$ is invertible (Lemma~\ref{thm:ls_inv}), by~\eqref{eq:Mhateta->M0} and because the map from the entries of a matrix to its eigenvalues is continuous, there is $\bar{N}$ such that $M(\hat{\eta}_N)$ is invertible for all $N>\bar{N}$, w.p.1.

Returning to~\eqref{eq:norm_thetals-theta0}, we may now write
\begin{equation}
\begin{aligned}
\big|\hspace{-1pt}\big|\thetals-\theta_\nul\big|\hspace{-1pt}\big| &\leq C \big|\hspace{-1pt}\big|\hat{\eta}_N-\hat{\eta}_\nul^{n(N)}\big|\hspace{-1pt}\big| , \quad \forall  N>\bar{N}. \\
&\to 0, \wpone.
\end{aligned}
\label{eq:thetals-thetanul_preparerate}
\end{equation}
Moreover, using~\eqref{eq:hateta-eta0n}, we can re-write~\eqref{eq:thetals-thetanul_preparerate} as
\begin{equation}
\big|\hspace{-1pt}\big|\thetals-\theta_\nul\big|\hspace{-1pt}\big| \leq C \big( \big|\hspace{-1pt}\big|\hat{\eta}_N-\bar{\eta}^{n(N)}\big|\hspace{-1pt}\big| + \big|\hspace{-1pt}\big|\bar{\eta}^{n(N)}-\eta^{n(N)}_\nul\big|\hspace{-1pt}\big| \big) .
\end{equation}
From Proposition~\ref{lemma:etabar-eta0}, we have ${||\bar{\eta}^{n(N)}-\eta^{n(N)}_\nul|| \leq C d(N)}$, which thus approaches zero faster than $||\hat{\eta}_N-\bar{\eta}^{n(N)}||$, whose decay rate is according to~\eqref{eq:hateta-bareta_rate}.
For the decay rate of $||\thetals-\theta_\nul||$, it suffices then to take the rate of the slowest-decaying term, which is given by~\eqref{eq:thetalsdecay}, as we wanted to show.
\hfill $\qed$

\section{Consistency of Step 3}

The main purpose of this appendix is to prove Theorem~\ref{thm:consistencyWLS}.
However, before we do so, we introduce some results regarding the norm of some vectors and matrices.

\begin{itemize}[leftmargin=*]
\item $\norm{R^n_N}$ \textit{is bounded for all $n$ and sufficiently large $N$, w.p.1}
\end{itemize}
Let $R^n_N$ be defined as in~\eqref{eq:RnN}.
Then, from Lemma 4.2 in~\cite{ljung&wahlberg92}, we have that there exists $\bar{N}$ such that, w.p.1,
\begin{equation}
\norm{R^n_N} \leq C, \;\; \forall n, \; \forall N>\bar{N}.
\label{eq:RnN_bounded}
\end{equation}

\begin{itemize}[leftmargin=*]
\item $||T^{\minus 1}_n(\theta_\nul)||$ \textit{is bounded for all $n$}
\end{itemize}
We observe that, with $T_n(\theta)$ is given by~\eqref{eq:T0}, the inverse of $T_n(\theta_\nul)$ is given by
\begin{equation}
T^{\minus 1}_n(\theta) =
\begin{bmatrix}
T^c_n(\theta)^{\minus 1} & 0 \\
T^f_n(\theta)^{\minus 1} T^l_n(\theta) T^c_n(\theta)^{\minus 1} & T^f_c(\theta)^{\minus 1}
\end{bmatrix} ,
\label{eq:Tinv}
\end{equation}
evaluated at the true parameters $\theta_\nul$.
Also, $T^f_n(\theta_\nul)^{\minus 1}$, $T^c_n(\theta_\nul)^{\minus 1}$, and $T^l_n(\theta_\nul)$ are sub-matrices of $\mathbb{T}[1/F_\nul(q)]$, $\mathbb{T}[1/C_\nul(q)]$, and $\mathbb{T}[L_\nul(q)]$, respectively, where $\mathbb{T}[X(q)]$ is defined by~\eqref{eq:mathbbT}.
Then, 
\begin{align}
\big|\hspace{-1pt}\big| T^{\minus 1}_n(\theta_\nul) \big|\hspace{-1pt}\big| \leq &\big|\hspace{-1pt}\big| T^f_n(\theta_\nul)^{\minus 1}\big|\hspace{-1pt}\big| + \big|\hspace{-1pt}\big| T^c_n(\theta_\nul)^{\minus 1}\big|\hspace{-1pt}\big| \\ 
&+ \big|\hspace{-1pt}\big| T^f_n(\theta_\nul)^{\minus 1}\big|\hspace{-1pt}\big| \, \big|\hspace{-1pt}\big| T^l_n(\theta_\nul)\big|\hspace{-1pt}\big| \, \big|\hspace{-1pt}\big|T^c_n(\theta_\nul)^{\minus 1}\big|\hspace{-1pt}\big| \\
\leq &\norm{\mathbb{T}[1/F_\nul(q)]} + \norm{\mathbb{T}[1/C_\nul(q)]} \\
&+ \norm{\mathbb{T}[1/F_\nul(q)]}\norm{\mathbb{T}[L_\nul(q)]}\norm{\mathbb{T}[1/C_\nul(q)]} \\
\leq &C , \;\; \forall n, \label{eq:Tinv0_bounded}
\end{align}
where the last inequality follows from~\eqref{eq:Tinfbounded} and the fact that $1/F_\nul(q)$, $1/C_\nul(q)$, and $L_\nul(q)$ are stable transfer functions.

\begin{itemize}[leftmargin=*]
\item $||T^{\minus 1}_n(\thetals)||$ \textit{is bounded for all $n$ and sufficiently large $N$}
\end{itemize}
Consider the term $||T^{\minus 1}_n(\thetals)||$, where $T^{\minus 1}_n(\thetals)$ is given by~\eqref{eq:Tinv} evaluated at $\thetals$.
We have that, proceeding as in~\eqref{eq:Tinv0_bounded},
\begin{equation}
\begin{aligned}
&\big|\hspace{-1pt}\big|T^{\minus 1}_n(\thetals)\big|\hspace{-1pt}\big| \leq \big|\hspace{-1pt}\big|\mathbb{T}[1/C(q,\thetals)]\big|\hspace{-1pt}\big| + \big|\hspace{-1pt}\big|\mathbb{T}[1/F(q,\thetals)]\big|\hspace{-1pt}\big| \\ 
&\hphantom{\leq}+ \big|\hspace{-1pt}\big|\mathbb{T}[1/C(q,\thetals)]\big|\hspace{-1pt}\big| \, \big|\hspace{-1pt}\big|\mathbb{T}[L(q,\thetals)]\big|\hspace{-1pt}\big| \, \big|\hspace{-1pt}\big|\mathbb{T}[1/F(q,\thetals)]\big|\hspace{-1pt}\big|
\end{aligned}
\label{eq:invTls_triang}
\end{equation}
for all $n$. This will be bounded if $F(q,\thetals)$ and $C(q,\thetals)$ have all poles strictly inside the unit circle.
From Theorem~\ref{thm:consistencyLS} and stability of the true system by Assumption~\ref{ass:model_truesystem}, we conclude that there exists $\bar{N}$ such that $F(q,\thetals)$ and $C(q,\thetals)$ have all roots strictly inside the unit circle for all $N>\bar{N}$.
Thus, we have
\begin{equation}
\big|\hspace{-1pt}\big|T^{\minus 1}_n(\thetals)\big|\hspace{-1pt}\big| \leq C , \;\; \forall n, \; \forall N>\bar{N}, \; \text{w.p.1} .
\label{eq:invTls_bounded}
\end{equation}

\begin{itemize}[leftmargin=*]
\item $||T^{\minus 1}_n(\hat{\theta}^\text{LS}_N) - T^{\minus 1}_n(\theta_\nul)||$ \textit{tends to zero, as $N$ tends to infinity, w.p.1}
\end{itemize}
For the term $|| T^{\minus 1}_n(\hat{\theta}^\text{LS}_N) - T^{\minus 1}_n(\theta_\nul) ||$, with $n=n(N)$,
we have
\begin{align}
\begin{multlined}
\big|\hspace{-1pt}\big|T^{\minus 1}_n(\hat{\theta}^\text{LS}_N) - T^{\minus 1}_n(\theta_\nul)\big|\hspace{-1pt}\big| \leq \\
\leq \big|\hspace{-1pt}\big|T^{\minus 1}_n(\hat{\theta}^\text{LS}_N)\big|\hspace{-1pt}\big| \, \big|\hspace{-1pt}\big|T_n(\hat{\theta}^\text{LS}_N)-T_n(\theta_\nul)\big|\hspace{-1pt}\big| \, \big|\hspace{-1pt}\big|T^{\minus 1}_n(\theta_\nul)\big|\hspace{-1pt}\big| .
\end{multlined}
\label{eq:invTls-invT0}
\end{align}
Because $||X||\leq\sqrt{||X||_1||X||_\infty}$, where $X$ is an arbitrary matrix, we have that
\begin{equation}
\begin{aligned}
\big|\hspace{-1pt}\big|T_n(\hat{\theta}^\text{LS}_N)-T_n(\theta_\nul)\big|\hspace{-1pt}\big| &\leq \textstyle{\sum_{k=1}^{m_f+m_l+m_c}}|\hat{\theta}^{k,\text{LS}}_N - \theta_\nul^k| \\
&\leq C \big|\hspace{-1pt}\big|\hat{\theta}^\text{LS}_N-\theta_\nul\big|\hspace{-1pt}\big| , 
\end{aligned}
\label{eq:Tls-T0}
\end{equation}
with superscript $k$ denoting the $k^\text{th}$ element of the vector, we can use Theorem~\ref{thm:consistencyLS} to show that
\begin{equation}
\big|\hspace{-1pt}\big|T_n(\hat{\theta}^\text{LS}_N)-T_n(\theta_\nul)\big|\hspace{-1pt}\big| = \mathcal{O} \bigg( \sqrt{n(N)\frac{\log N}{N}} \big(1+d(N)\big) \bigg).
\label{eq:O_That-T0}
\end{equation}
From~\eqref{eq:D3} and~\eqref{eq:D4},
\begin{equation}
\sqrt{n^2(N)\frac{\log N}{N}}\big(1+d(N)\big) \to 0, \text{ as } N\to\infty ,
\end{equation}
and thus  
\begin{equation}
\big|\hspace{-1pt}\big|T_n(\hat{\theta}^\text{LS}_N)-T_n(\theta_\nul)\big|\hspace{-1pt}\big| \rightarrow 0, \wpone .
\end{equation}
Together with~\eqref{eq:Tinv0_bounded}, \eqref{eq:invTls_bounded}, and~\eqref{eq:invTls-invT0}, this implies that
\begin{equation}
\big|\hspace{-1pt}\big|T^{\minus 1}_n(\hat{\theta}^\text{LS}_N) - T^{\minus 1}_n(\theta_\nul)\big|\hspace{-1pt}\big| \rightarrow 0 , \wpone.
\label{eq:Tinvhat-Tinv0->0}
\end{equation}

The following two lemmas are useful for the invertibility of the weighted least squares problem~\eqref{eq:theta_wls}.

\begin{lemma}
\label{thm:wls_inv}
Let Assumption~\ref{ass:model_truesystem} hold and
\begin{equation}
\bar{M}(\eta_\nul,\theta_\nul) := \lim\limits_{n\to\infty} Q_n^\top(\eta^n_\nul) \bar{W}_n(\theta_\nul) Q_n(\eta^n_\nul),
\label{eq:M_etanuln_theta0}
\end{equation}
where $\bar{W}_n(\theta_\nul)$ is given by~\eqref{eq:barW}, and $Q_n(\eta^n_\nul)$ is defined by~\eqref{eq:Qeta} at the true parameters $\eta^n_\nul$.
Then, $\bar{M}(\eta_\nul,\theta_\nul)$ is invertible.
\end{lemma}

\begin{proof}
Using~\eqref{eq:Rbar} and~\eqref{eq:barW}, we re-write~\eqref{eq:M_etanuln_theta0} as
\begin{multline}
\bar{M}(\eta_\nul,\theta_\nul) = \\ \lim\limits_{n\to\infty} Q^\top_n(\eta^n_\nul) T^{\minus\top}_n(\theta_\nul) \bar{\mathbb{E}} \left[ \varphi_t^n (\varphi_t^n)^\top \right] T^{\minus 1}_n(\theta_\nul) Q_n(\eta^n_\nul) .
\label{eq:barM_rewritten}
\end{multline}
Re-writing $\varphi_t^n$, defined in~\eqref{eq:phi}, as
\begin{equation}
\varphi_t^n =
\begin{bmatrix}
-\Gamma_n y_t \\ \Gamma_n u_t
\end{bmatrix}
=
\begin{bmatrix}
-\Gamma_n G_\nul(q) & -\Gamma_n H_\nul(q) \\
\Gamma_n & 0
\end{bmatrix}
\begin{bmatrix}
u_t \\ e_t
\end{bmatrix} ,
\end{equation}
we can then write
\begin{equation}
\bar{\mathbb{E}} \left[ \varphi_t^n (\varphi_t^n)^\top \right] =
\frac{1}{2\pi} \int_{\minus\pi}^{\pi} \Lambda_n(e^{i\omega}) \Phi_z \Lambda_n^*(e^{i\omega}) d\omega ,
\end{equation}
where
\begin{equation}
\Lambda_n(q) =
\begin{bmatrix}
-\Gamma_n G_\nul(q) & -\Gamma_n H_\nul(q) \\
\Gamma_n & 0
\end{bmatrix} .
\end{equation}
Then, we can re-write~\eqref{eq:barM_rewritten} as
\begin{multline}
\bar{M}(\eta_\nul,\theta_\nul) = \frac{1}{2\pi} \int_{\minus\pi}^{\pi} \lim\limits_{n\to\infty} Q^\top_n(\eta^n_\nul) T^{\minus\top}_n(\theta_\nul)  \Lambda_n(e^{i\omega}) \\ \cdot \Phi_z \Lambda_n^*(e^{i\omega}) T^{\minus 1}_n(\theta_\nul) Q_n(\eta^n_\nul) d\omega .
\label{eq:M0_freqdom}
\end{multline}
Moreover,
\begin{equation}
Q_n^\top(\eta^n_\nul)T^{\minus\top}_n(\theta_\nul) =
\begin{bmatrix}
0 & -\mathcal{T}_{n,m_f}^\top\left(\frac{B_\nul}{F_\nul}\right) \\
0 & \mathcal{T}_{n,m_l}^\top\left(\frac{A_\nul}{F_\nul}\right) \\
-\mathcal{T}_{n,m_c}^\top\left(\frac{A_\nul}{C_\nul}\right) & -\mathcal{T}_{n,m_c}^\top\left(\frac{L_\nul A_\nul}{F_\nul C_\nul}\right) \\
\mathcal{T}_{n,m_d}^\top\left(\frac{1}{C_\nul}\right) & \mathcal{T}_{n,m_d}^\top\left(\frac{L_\nul}{F_\nul C_\nul}\right)
\end{bmatrix} ,
\end{equation}
where the argument $q$ of the polynomials was dropped for notational simplicity.
In turn, we can also write
\begin{equation}
\begin{aligned}
&Q_n^\top(\eta^n_\nul)T^{\minus\top}_n(\theta_\nul)\Lambda_n = \\
&\begin{bmatrix}
-\mathcal{T}_{n,m_f}^\top\!\!\left(\!\frac{B_\nul}{F_\nul}\!\right) \!\! \Gamma_n & 0\\
\mathcal{T}_{n,m_l}^\top\!\!\left(\!\frac{A_\nul}{F_\nul}\!\right) \!\! \Gamma_n & 0 \\
\mathcal{T}_{n,m_c}^\top\!\!\left(\!\frac{A_\nul}{C_\nul}\!\right) \!\! \Gamma_n G_\nul \!\!-\!\! \mathcal{T}_{n,m_c}^\top\!\!\left(\!\frac{L_\nul A_\nul}{F_\nul C_\nul}\!\right) \!\! \Gamma_n & \mathcal{T}_{n,m_c}^\top\!\!\left(\!\frac{A_\nul}{C_\nul}\!\right) \!\! \Gamma_n H_\nul \\
-\mathcal{T}_{n,m_d}^\top\!\!\left(\!\frac{1}{C_\nul}\!\right) \!\! \Gamma_n G_\nul \!\!+\!\! \mathcal{T}_{n,m_d}^\top\!\!\left(\!\frac{L_\nul}{F_\nul C_\nul}\!\right) \!\! \Gamma_n & -\mathcal{T}_{n,m_d}^\top\!\!\left(\!\frac{1}{C_\nul}\!\right) \!\! \Gamma_n H_\nul
\end{bmatrix} .
\end{aligned}
\end{equation}
It is possible to observe that, for some polynomial ${X(q)=\sum_{k=0}^\infty x_k q^{\minus k}}$, $\lim\limits_{n\to\infty} \mathcal{T}_{n,m}^\top (X(q)) \Gamma_n = X(q) \Gamma_m$.
Then, using also~\eqref{eq:truearxpoly}, we have $\lim\limits_{n\to\infty} Q_n^\top(\eta^n_\nul)T^{\minus\top}_n(\theta_\nul)\Lambda_n = \Omega$, where $\Omega$ is given by~\eqref{eq:Omega}. This allows us to re-write~\eqref{eq:M0_freqdom} as
\begin{equation}
\bar{M}(\eta_\nul,\theta_\nul) = \frac{1}{2\pi} \int_{\minus\pi}^{\pi} \Omega \Phi_z \Omega^* d\omega = M_\text{CR} ,
\label{eq:barM=MCR}
\end{equation}
which is invertible because the CR bound exists for an informative experiment~\cite{ljung99}.
\end{proof}

\begin{lemma}
\label{lem:wls_inv_stoch}
Let Assumptions~\ref{ass:model_truesystem}, \ref{ass:noise}, \ref{ass:input}, and~\ref{ass:ARXorder} hold.
Also, let $M(\hat{\eta}_N,\thetals) := Q_n^\top(\hat{\eta}_N) W_n(\thetals) Q_n(\hat{\eta}_N)$, where ${\hat{\eta}_N:=\hat{\eta}^{n(N)}_N}$ is defined by~\eqref{eq:eta_regls}, $Q_n(\hat{\eta}_N)$ is defined by~\eqref{eq:Qeta} evaluated at the estimated parameters $\hat{\eta}_N$, and $W_n(\thetals)$ is defined by~\eqref{eq:Wls}.
Then, 
\begin{equation}
M(\hat{\eta}_N,\thetals) \to \bar{M}(\eta_\nul,\theta_\nul), \wpone .
\label{eq:M0->limnM}
\end{equation}
\end{lemma}

\begin{proof}
For the purpose of showing~\eqref{eq:M0->limnM}, we will show that
\begin{equation}
\begin{multlined}
\big|\hspace{-1pt}\big|M(\hat{\eta}_N,\thetals) - Q_n^\top(\eta^{n(N)}_\nul) \bar{W}_n(\theta_\nul) Q_n(\eta^{n(N)}_\nul)\big|\hspace{-1pt}\big| \\ \to 0, \wpone.
\end{multlined}
\label{eq:normQWQhat-QWQ}
\end{equation}
We apply Proposition~\ref{prop:Xdiff->0}, whose conditions can be verified using~\eqref{eq:Qhateta-Qeta0}, \eqref{eq:bound_Qhateta}, \eqref{eq:RnN_bounded}, and~\eqref{eq:invTls_bounded}, but additionally we need to verify that $||\bar{W}_n(\theta_\nul)||$ is bounded and $||W_n(\hat{\theta}^\text{LS}_N)-\bar{W}_n(\theta_\nul)||$ tends to zero.
For the first, we have that
\begin{equation}
||\bar{W}_n(\theta_\nul)|| \leq ||T^{-1}_n(\theta_\nul)||^2 ||\bar R^n|| \leq C,
\label{eq:barWnbound}
\end{equation}
following from~\eqref{eq:RnN_bounded} and~\eqref{eq:Tinv0_bounded}.
For the second, with $W_n(\thetals)$ given by~\eqref{eq:Wls} and $\bar{W}_n(\theta_\nul)$ by~\eqref{eq:barW}, conditions in Proposition~\ref{prop:Xdiff->0} are satisfied using~\eqref{eq:RnN_bounded}, \eqref{eq:Tinv0_bounded}, \eqref{eq:Tinvhat-Tinv0->0} and Proposition~\ref{lemma:RnN-barRn}, from where it follows that 
\begin{equation}
||W_n(\thetals)-\bar{W}_n(\theta_\nul)|| \to 0, \wpone .
\label{eq:hatWn-barWn->0}
\end{equation}
Having shown~\eqref{eq:barWnbound} and~\eqref{eq:hatWn-barWn->0}, the assumptions of Proposition~\ref{prop:Xdiff->0} are verified, from which~\eqref{eq:normQWQhat-QWQ} follows and implies~\eqref{eq:M0->limnM}.
\end{proof}

We now have the necessary results to prove Theorem~\ref{thm:consistencyWLS}.

\subsubsection*{Proof of Theorem~\ref{thm:consistencyWLS}}
\label{app:wls_consistency_proof}

Similarly to~\eqref{eq:thetals-theta0}, we write
\begin{multline}
\thetawls-\theta_\nul \\ = M^{\minus 1}(\hat{\eta}_N,\thetals) Q_n^\top(\hat{\eta}_N) W_n(\thetals) T_n(\theta_\nul) (\hat{\eta}_N-\eta^{n(N)}_\nul) ,
\label{eq:thetawls-theta0_expand}
\end{multline}
and analyze
\begin{multline}
\big|\hspace{-1pt}\big|\thetawls-\theta_\nul\big|\hspace{-1pt}\big| \leq \big|\hspace{-1pt}\big|M^{\minus 1}(\hat{\eta}_N,\thetals)\big|\hspace{-1pt}\big| \, \big|\hspace{-1pt}\big|Q_n(\hat{\eta}_N)\big|\hspace{-1pt}\big| \\ \cdot \big|\hspace{-1pt}\big|W_n(\thetals)\big|\hspace{-1pt}\big| \, \big|\hspace{-1pt}\big|T_n(\theta_\nul)\big|\hspace{-1pt}\big| \, \big|\hspace{-1pt}\big|\hat{\eta}_N-\eta^{n(N)}_\nul\big|\hspace{-1pt}\big| .
\label{eq:norm_thetawls-theta0}
\end{multline}
From Lemma~\ref{lem:wls_inv_stoch}, $M(\hat{\eta}_N,\thetals)$ converges to $\bar{M}(\eta_\nul,\theta_\nul)$, which is invertible from Lemma~\ref{thm:wls_inv}. 
Hence, because the map from the entries of the matrix to its eigenvalues is continuous, $M(\hat{\eta}_N,\thetals)$ is invertible for sufficiently large $N$, and therefore its norm is bounded, as it is a matrix of fixed dimensions.
Also, from~\eqref{eq:bound_Qhateta}, $\norm{Q_n(\hat{\eta}_N)}$ is bounded for sufficiently large $N$.
Moreover, we have that, making explicit that $n=n(N)$,
\begin{equation}
\big|\hspace{-1pt}\big|W_{n(N)}(\thetals)\big|\hspace{-1pt}\big| \leq \big|\hspace{-1pt}\big|T^{\minus 1}_{n(N)}(\thetals)\big|\hspace{-1pt}\big|^2 \, \big|\hspace{-1pt}\big|R^{n(N)}_N\big|\hspace{-1pt}\big| .
\label{eq:normWthetals}
\end{equation}
Then, from~\eqref{eq:invTls_bounded} and~\eqref{eq:RnN_bounded}, we have
\begin{equation}
\big|\hspace{-1pt}\big|W_{n(N)}(\thetals)\big|\hspace{-1pt}\big| \leq C, \; \forall N>\bar{N} .
\label{eq:Wls_bounded}
\end{equation}
Finally, using also~\eqref{eq:bound_Qeta0}, \eqref{eq:T0_bounded}, and~\eqref{eq:hateta-eta0}, we conclude that
\begin{equation}
\big|\hspace{-1pt}\big|\thetawls-\theta_\nul\big|\hspace{-1pt}\big| \to 0, \wpone. \qed
\end{equation}

\section{Asymptotic Distribution and Covariance of Step 3}
\label{app:asymp_cov_proof}

The purpose of this appendix is to prove Theorem~\ref{thm:asycov}: asymptotic distribution and covariance of
\begin{equation}
\sqrt{N} (\thetawls-\theta_\nul) = \sqrt{N} \Upsilon^n(\hat{\eta}_N,\thetals) (\hat{\eta}_N-\eta^{n(N)}_\nul) ,
\label{eq:sqrtN_thetadiff_1}
\end{equation}
which is re-written from~\eqref{eq:thetawls-theta0_expand}, where 
\begin{equation}
\begin{multlined}
\Upsilon^n(\hat{\eta}_N,\thetals) = [Q_n^\top(\hat{\eta}_N) W_n(\thetals) Q_n(\hat{\eta}_N)]^{-1} \\ \cdot Q_n^\top(\hat{\eta}_N) W_n(\thetals) T_n(\theta_\nul).
\end{multlined}
\end{equation}
If $\Upsilon^n(\hat{\eta}_N,\thetals)$ were of fixed dimensions, the standard idea would be to show that $\Upsilon^n(\hat{\eta}_N,\thetals)$ converges w.p.1 to a deterministic matrix, as consequence of $\hat\eta_N$ and $\thetals$ being consistent estimates of $\eta_\nul$ and $\theta_\nul$, respectively.
Then, for computing the asymptotic distribution and covariance of~\eqref{eq:sqrtN_thetadiff_1}, one can consider the asymptotic distribution and covariance of $\sqrt{N}(\hat{\eta}_N-\eta^{n(N)}_\nul)$ while $\Upsilon^n(\hat{\eta}_N,\thetals)$ can be replaced by the deterministic matrix it converges to.
This standard result follows from~\cite[Lemma B.4]{soderstromstoica89book}, but it is not applicable here because the dimensions of $\Upsilon^n(\hat{\eta}_N,\thetals)$ and $\hat{\eta}_N-\eta^{n(N)}_\nul$ are not fixed.
In this scenario, Proposition~\ref{thm:Upsilon} must be used instead.
However, \eqref{eq:sqrtN_thetadiff_1} is not ready to be used with Proposition~\ref{thm:Upsilon}, because it requires $\hat{\eta}_N-\eta^n_\nul$ to be pre-multiplied by a deterministic matrix.
The key idea of proving Theorem~\ref{thm:asycov} is to show that~\eqref{eq:sqrtN_thetadiff_1} has the same asymptotic distribution and covariance as an expression of the form $\bar \Upsilon^n [\hat{\eta}_N-\eta^{n(N)}_\nul]$, where $\bar \Upsilon^n$ is a deterministic matrix, and then apply Proposition~\ref{thm:Upsilon}.
The following result will be useful for this purpose.

\begin{prop}
Let $\hat{x}_N = \sqrt{N} \hat{A}_N \hat{B}_N \hat{\delta}_N $ be a finite-dimensional vector, where $\hat{A}_N$ and $\hat{B}_N$ are stochastic matrices and $ \hat{\delta}_N $ is a stochastic vector of compatible dimensions.
The dimensions may increase to infinity as function of $N$, except for the number of rows of $\hat{A}_N$, which is fixed.
We assume that there is $\bar{N}$ such that $||\hat{A}_N||<C$ for all $N>\bar{N}$, there is $\bar{B}$ such that $||\hat{B}_N-\bar{B}||\to 0$ as $N\to\infty$ w.p.1, and $||\hat{\delta}_N||\to 0$ as $N\to\infty$ w.p.1.
Then, if $\sqrt{N}||\hat{B}_N-\bar{B}||\,||\hat{\delta}_N||\to 0$ as $N\to\infty$ w.p.1, $\hat{x}_N$ and $\sqrt{N} \hat{A}_N \bar{B} \hat{\delta}_N$ have the same asymptotic distribution and covariance.
\label{prop:SSlemma}
\end{prop}

\begin{proof}
We begin by writing
\begin{equation}
\hat{x}_N = \sqrt{N} \hat{A}_N \bar{B} \hat{\delta}_N + \sqrt{N} \hat{A}_N (\hat{B}_N - \bar{B}) \hat{\delta}_N .
\label{eq:hatx_expand}
\end{equation}
Although some of the matrix and vector dimensions may increase to infinity with $N$, the number of rows of $\hat{A}_N$ is fixed, which makes $\hat{x}_N$ finite dimensional, to which~\cite[Lemma B.4]{soderstromstoica89book} may be applied.
Then, it is a consequence of this lemma that $\hat{x}_N$ and $\sqrt{N}\hat{A}_N \bar{B} \hat{\delta}_N$ have the same asymptotic distribution and covariance if the second term on the right side of~\eqref{eq:hatx_expand} tends to zero with probability one.
By assumption, we have
\begin{equation}
\begin{multlined}
|| \sqrt{N} \hat{A}_N (\hat{B}_N - \bar{B}) \hat{\delta}_N || \leq
\sqrt{N} \, ||\hat{A}_N|| \, ||\hat{B}_N - \bar{B}|| \, ||\hat{\delta}_N|| \\
\to 0, \wpone,
\end{multlined}
\end{equation}
which completes the proof.
\end{proof}

We now have the necessary results to prove Theorem~\ref{thm:asycov}.
\subsubsection*{Proof of Theorem~\ref{thm:asycov}}
We start by re-writing~\eqref{eq:sqrtN_thetadiff_1} as
\begin{equation}
\sqrt{N} (\thetawls-\theta_\nul) = M^{\minus 1}(\hat{\eta}_N,\thetals) x(\hat{\eta}_N,\thetals) ,
\end{equation}
where 
\begin{equation}
\begin{aligned}
M(\hat{\eta}_N,\thetals) &=  Q_n^\top(\hat{\eta}_N) W_n(\thetals) Q_n(\hat{\eta}_N) , \\
x(\hat{\eta}_N,\thetals) &=  \sqrt{N} Q_n^\top(\hat{\eta}_N) W_n(\thetals) T_n(\theta_\nul) (\hat{\eta}_N\!-\!\eta^{n(N)}_\nul) .
\end{aligned}
\label{eq:Mx}
\end{equation}
Both $M(\hat{\eta}_N,\thetals)$ and $x(\hat{\eta}_N,\thetals)$ are of fixed dimension, and we have from~\eqref{eq:barM=MCR} and~\eqref{eq:M0->limnM} that
\begin{equation}
M^{\minus 1}(\hat{\eta}_N,\thetals) \to M_\text{CR}^{\minus 1} , \wpone .
\end{equation}
Then, if we assume that
\begin{equation}
x(\hat{\eta}_N,\thetals) \sim As \mathcal{N} (0,P) ,
\label{eq:xdist}
\end{equation}
we have that, from~\cite[Lemma B.4]{soderstromstoica89book},
\begin{equation}
\sqrt{N}(\thetawls-\theta_\nul) \sim As\mathcal{N} \bigg(0,M_\text{CR}^{\minus 1} P M_\text{CR}^{\minus 1} \bigg) .
\label{eq:thetawls_dist_MandP}
\end{equation}
We will proceed to show that~\eqref{eq:xdist} is verified with
\begin{equation}
P = \sigma_\nul^2 \lim\limits_{n\to\infty} Q_n^\top(\eta^n_\nul) \bar{W}_n(\theta_\nul) Q_n(\eta^n_\nul) = \sigma_\nul^2 M_\text{CR},
\label{eq:P}
\end{equation}
where the second equality follows directly from~\eqref{eq:M_etanuln_theta0} and~\eqref{eq:barM=MCR}. 
We now proceed to show the first equality.

In the following arguments, we will apply Proposition~\ref{prop:SSlemma} repeatedly to $x(\hat{\eta}_N,\thetals)$. 
This required the boundedness of some matrices; however, because all the matrices in $x(\hat{\eta}_N,\thetals)$ have been shown to be bounded for sufficiently large $N$ w.p.1, for readability we will refrain from referring to this every time Proposition~\ref{prop:SSlemma} is applied.

Because it is more convenient to work with $\bar{\eta}^{n(N)}$ than $\eta_\nul^{n(N)}$, we start by re-writing $x(\hat{\eta}_N,\thetals)$ as
\begin{equation}
\begin{multlined}
x(\hat{\eta}_N,\thetals) = \sqrt{N} Q_n^\top(\hat{\eta}_N) W_n(\hat{\theta}^\text{LS}_N) T_n(\theta_\nul) (\hat{\eta}_N-\bar{\eta}^{n(N)})\\
+ Q_n^\top(\hat{\eta}_N) W_n(\hat{\theta}^\text{LS}_N) T_n(\theta_\nul) \sqrt{N} (\bar{\eta}^{n(N)}-\eta^{n(N)}_\nul) .
\end{multlined}
\end{equation}
Using Proposition~\ref{lemma:etabar-eta0}, we have, for sufficiently large $N$ w.p.1,
\begin{align}
\begin{split}
\big|\hspace{-1pt}\big|Q_n^\top(\hat{\eta}_N) W_n(\hat{\theta}^\text{LS}_N) T_n(\theta_\nul) \sqrt{N} (\bar{\eta}^{n(N)}-\eta^{n(N)}_\nul)\big|\hspace{-1pt}\big| \leq& \\
\leq C \sqrt{N} d(N) \rightarrow 0, \text{ as } N\rightarrow\infty .&
\end{split}
\end{align}
Using an identical argument to Proposition~\ref{prop:SSlemma}, we have that $x(\hat{\eta}_N,\thetals)$ and
\begin{equation}
\sqrt{N} Q_n^\top(\hat{\eta}_N) W_n(\hat{\theta}^\text{LS}_N) T_n(\theta_\nul) (\hat{\eta}_N-\bar{\eta}^{n(N)})
\label{eq:-sqrtNV'_gfixed}
\end{equation}
have the same asymptotic distribution and covariance, so we will analyze~\eqref{eq:-sqrtNV'_gfixed} instead.

Expanding $W_n(\hat{\theta}^\text{LS}_N)$ in~\eqref{eq:-sqrtNV'_gfixed}, we obtain
\begin{equation}
\begin{multlined}
\sqrt{N} Q_n^\top(\hat{\eta}_N) W_n(\hat{\theta}^\text{LS}_N) T_n(\theta_\nul) (\hat{\eta}_N-\bar{\eta}^{n(N)}) \\
\!\!= \!\sqrt{N} Q_n^\top\!(\hat{\eta}_N) T^{-\top}_n\!(\hat{\theta}^\text{LS}_N) R^n_N T^{-1}_n\!(\hat{\theta}^\text{LS}_N) T_n(\theta_\nul) (\hat{\eta}_N\!-\bar{\eta}^{n(N)}).
\end{multlined}
\label{eq:-sqrtNV'_expandW}
\end{equation}
Using Proposition~\ref{prop:SSlemma}, we conclude that \eqref{eq:-sqrtNV'_expandW} and
\begin{equation}
\begin{multlined}
\sqrt{N} Q_n^\top\!(\hat{\eta}_N) T^{-\top}_n\!(\hat{\theta}^\text{LS}_N) R^n_N T^{-1}_n\!(\theta_\nul) T_n(\theta_\nul) (\hat{\eta}_N\!-\bar{\eta}^{n(N)}) \\
= \sqrt{N} Q_n^\top (\hat{\eta}_N) T^{-\top}_n (\hat{\theta}^\text{LS}_N) R^n_N (\hat{\eta}_N\!-\bar{\eta}^{n(N)})
\end{multlined}
\label{eq:-sqrtNV'_T1fixed}
\end{equation}
have the same asymptotic properties if
\begin{equation}
\begin{multlined}
\sqrt{N}\,||T_n^{-1}(\thetals)-T_n^{-1}(\theta_\nul)||\,||\hat{\eta}_N\!-\bar{\eta}^{n(N)}||\to 0, \\ \wpone .
\end{multlined}
\end{equation}
Using~\eqref{eq:Tinv0_bounded}, \eqref{eq:invTls_bounded} and~\eqref{eq:invTls-invT0} to write
\begin{equation}
\begin{multlined}
\sqrt{N}\,||T_n^{-1}(\thetals)-T_n^{-1}(\theta_\nul)||\,||\hat{\eta}_N-\bar{\eta}^{n(N)}|| \\
\qquad \leq C \sqrt{N}\,||T_n(\thetals)-T_n(\theta_\nul)|| \,||\hat{\eta}_N-\bar{\eta}^{n(N)}|| ,
\end{multlined}
\end{equation}
we have from~\eqref{eq:O_That-T0} and Proposition~\ref{thm:hateta-bareta} that
\begin{multline}
\sqrt{N} \big|\hspace{-1pt}\big|T_n(\theta_\nul)-T_n(\hat{\theta}^\text{LS}_N) \big|\hspace{-1pt}\big| \, \big|\hspace{-1pt}\big|\hat{\eta}_N-\bar{\eta}^{n(N)}\big|\hspace{-1pt}\big| = \\
= \mathcal{O} \bigg( \frac{n(N)\log N}{\sqrt{N}} \big(1+d(N)\big)^2  \bigg) ,
\label{eq:O_sqrtN_Tdiff_gdiff}
\end{multline}
where
\begin{equation}
\frac{n(N)\log N}{\sqrt{N}} = \left(\frac{n^{3+\delta}(N)}{N}\right)^{\frac{1}{3+\delta}} \frac{\log N}{N^{\frac{1+\delta}{2(3+\delta)}}}
\rightarrow 0, \text{ as } N\rightarrow\infty ,
\label{eq:nlogN/sqrtN->0}
\end{equation}
due to Condition D2 in Assumption~\ref{ass:ARXorder}.
This implies that \eqref{eq:-sqrtNV'_expandW} and~\eqref{eq:-sqrtNV'_T1fixed}, and in turn $x(\hat{\eta}_N,\thetals)$, have the same asymptotic distribution and covariance.
Repeating this procedure, it can be shown that these, in turn, have the same asymptotic distribution and covariance as
\begin{equation}
\sqrt{N} Q_n^\top (\hat{\eta}_N) T^{-\top}_n (\theta_\nul) R^n_N (\hat{\eta}_N-\bar{\eta}^{n(N)}).
\label{eq:-sqrtNV'_Tfixed}
\end{equation}

There are two stochastic matrices left in~\eqref{eq:-sqrtNV'_Tfixed}, which we need to replace by deterministic matrices that do not affect the asymptotic properties.
Using Proposition~\ref{lemma:RnN-barRn}, 
\begin{align}
\begin{split}
&\sqrt{N} \big|\hspace{-1pt}\big|R^n_N - \bar{R}^n\big|\hspace{-1pt}\big| \, \big|\hspace{-1pt}\big|\hat{\eta}_N-\bar{\eta}^{n(N)}\big|\hspace{-1pt}\big| = \\
&\begin{multlined}[.9\displaywidth]
= \mathcal{O} \bigg( 2 \frac{n^{3/2}(N) \log N}{\sqrt{N}} \big(1+d(N)\big) \\
+ C \sqrt{\frac{n^2(N)\log N}{N}}\sqrt{\frac{n^3(N)}{N}}\big(1+d(N)\big) \bigg),
\end{multlined}
\end{split}
\end{align}
where the first term tends to zero by applying Condition D2 in Assumption~\ref{ass:ARXorder} to
\begin{multline}
\frac{n^{3/2}(N)\log N}{\sqrt{N}} = \left(\frac{n^{4+\delta}(N)}{N}\right)^{\frac{3}{2(4+\delta)}} \frac{\log N}{N^{\frac{1+\delta}{2(4+\delta)}}} \\
\rightarrow 0, \text{ as } N\rightarrow\infty ,
\label{eq:n3/2logN/sqrtN->0}
\end{multline}
and the second because of Condition D2 in Assumption~\ref{ass:ARXorder}, and~\eqref{eq:D3}.
Then, from Proposition~\ref{prop:SSlemma}, we have that \eqref{eq:-sqrtNV'_Tfixed} and
\begin{equation}
\sqrt{N} Q^\top_n(\hat{\eta}_N) T_n^{\minus \top}(\theta_\nul) \bar{R}^n (\hat{\eta}_N-\bar{\eta}^{n(N)}),
\label{eq:-sqrtNV'_Rfixed}
\end{equation}
and in turn $x(\hat{\eta}_N,\thetals)$, have the same asymptotic distribution and covariance, so we will analyze~\eqref{eq:-sqrtNV'_Rfixed}.

Applying again Proposition~\ref{prop:SSlemma}, we have that \eqref{eq:-sqrtNV'_Rfixed} and
\begin{equation}
\sqrt{N} Q^\top_n(\bar{\eta}^{n(N)}) T_n^{\minus \top}(\theta_\nul) \bar{R}^n (\hat{\eta}_N-\bar{\eta}^{n(N)})
\label{eq:-sqrtNV'_Qbar}
\end{equation}
have the same asymptotic properties, since
\begin{equation}
\begin{aligned}
\sqrt{N} &\big|\hspace{-1pt}\big|Q_n(\hat{\eta}_N) - Q_n(\bar{\eta}^{n(N)})\big|\hspace{-1pt}\big| \, \big|\hspace{-1pt}\big|\hat{\eta}_N-\bar{\eta}^{n(N)}\big|\hspace{-1pt}\big| \\
&\leq C \sqrt{N} \big|\hspace{-1pt}\big|\hat{\eta}_N-\bar{\eta}^{n(N)}\big|\hspace{-1pt}\big|^2 \\
&=\mathcal{O} \bigg( \frac{n(N)\log N}{\sqrt{N}}\big(1+d(N)\big)^2 \bigg)
\end{aligned}
\end{equation}
by using~\eqref{eq:Qhateta-Qeta0_proof} and Proposition~\ref{thm:hateta-bareta}, which tends to zero as ${N\to\infty}$, identically to~\eqref{eq:O_sqrtN_Tdiff_gdiff}.

In~\eqref{eq:-sqrtNV'_Qbar}, the matrix multiplying $\hat{\eta}_N-\bar{\eta}^{n(N)}$ is finally deterministic, but it will be more convenient to work with $Q(\eta_\nul^{n(N)})$.
With this purpose, Proposition~\ref{lemma:etabar-eta0} can be used to show that~\eqref{eq:-sqrtNV'_Qbar} and 
\begin{equation}
\sqrt{N} Q^\top_n(\eta^{n(N)}_\nul) T_n^{\minus \top}(\theta_\nul) \bar{R}^n (\hat{\eta}_N-\bar{\eta}^{n(N)})
\label{eq:-sqrtNV'_deterministic}
\end{equation}
have the same asymptotic properties, as
\begin{equation}
\begin{aligned}
\sqrt{N} &\big|\hspace{-1pt}\big| Q_n(\bar{\eta}^{n(N)}) - Q_n(\eta^n_\nul) \big|\hspace{-1pt}\big|\, \big|\hspace{-1pt}\big|\hat{\eta}_N-\bar{\eta}^{n(N)}\big|\hspace{-1pt}\big| \\
&\leq C \sqrt{N} \big|\hspace{-1pt}\big|\bar{\eta}^{n(N)}-\eta^{n(N)}_\nul\big|\hspace{-1pt}\big| \big|\hspace{-1pt}\big|\hat{\eta}_N-\bar{\eta}^{n(N)}\big|\hspace{-1pt}\big| \\
&=\mathcal{O} \bigg( \sqrt{\frac{n(N)\log N}{N}} \big(1+d(N)\big)\sqrt{N}d(N) \bigg),
\end{aligned}
\end{equation}
which tends to zero due to~\eqref{eq:D3} and~\eqref{eq:D4}.
Thus, $x(\hat{\eta}_N,\thetals)$ and~\eqref{eq:-sqrtNV'_deterministic} have the same asymptotic distribution and covariance, so we will analyze~\eqref{eq:-sqrtNV'_deterministic} instead.

Let $\Upsilon^n := Q^\top_n(\eta^n_\nul) T_n^{\minus \top}(\theta_\nul) \bar{R}^n$.
Then, using Proposition~\ref{thm:Upsilon},
\begin{equation}
\sqrt{N} \Upsilon^n (\hat{\eta}_N-\bar{\eta}^{n(N)}) \sim As\mathcal{N}(0,P) ,
\end{equation}
where $P$ is given by~\eqref{eq:P}.
Finally, using~\eqref{eq:barM=MCR}, \eqref{eq:thetawls_dist_MandP}, and~\eqref{eq:P}:
\begin{equation}
\sqrt{N} (\thetawls-\theta_\nul) \sim As\mathcal{N} (0,\sigma_\nul^2 M_\text{CR}^{\minus 1}).\qed
\end{equation}

\bibliographystyle{unsrt}        
\bibliography{mybib}             
                                 
\begin{IEEEbiography}{Miguel Galrinho}
was born in 1988.
He received his M.S. degree in aerospace engineering in 2013 from Delft University of Technology,
The Netherlands, and the Licentiate degree in electrical engineering in 2016
from KTH Royal Institute of Technology, Stockholm, Sweden.

He is currently a PhD student at KTH, with the Department of Automatic Control, School of Electrical Engineering,
under supervision of Professor H{\aa}kan Hjalmarsson.
His research is on least-squares methods for identification of structured models.
\end{IEEEbiography}

\begin{IEEEbiography}{Cristian R. Rojas}
(M'13) was born in 1980. He
received the M.S. degree in electronics engineering
from the Universidad T{\'e}cnica Federico Santa
Mar{\'i}a, Valpara{\'i}so, Chile, in 2004, and the Ph.D.
degree in electrical engineering at The University of
Newcastle, NSW, Australia, in 2008.

Since October 2008, he has been with the Royal
Institute of Technology, Stockholm, Sweden, where
he is currently Associate Professor at the Department of Automatic
Control, School of Electrical Engineering.
His research interests lie in system identification and
signal processing.
\end{IEEEbiography}

\begin{IEEEbiography}{H{\aa}kan Hjalmarsson}
(M'98--SM'11--F'13) was
born in 1962. He received the M.S. degree in electrical
engineering in 1988, and the Licentiate and
Ph.D. degrees in automatic control in 1990 and
1993, respectively, all from Link{\"o}ping University,
Link{\"o}ping, Sweden.

He has held visiting research positions at California
Institute of Technology, Louvain University,
and at the University of Newcastle, Australia. His research
interests include system identification, signal
processing, control and estimation in communication
networks and automated tuning of controllers.

Dr. Hjalmarsson has served as an Associate Editor for Automatica
(1996--2001) and for the IEEE Transactions on Automatic Control
(2005--2007), and has been Guest Editor for the European Journal of Control
and Control Engineering Practice. He is a Professor at the School of Electrical
Engineering, KTH, Stockholm, Sweden. He is a Chair of the IFAC Coordinating
Committee CC1 Systems and Signals. In 2001, he received the KTH award
for outstanding contribution to undergraduate education. He is co-recipient of
the European Research Council advanced grant.
\end{IEEEbiography}

\end{document}